\documentclass[a4paper,USenglish,cleveref,autoref,thm-restate,numberwithinsect,nameinlink]{lipics-v2021}

\hideLIPIcs
\nolinenumbers

\usepackage[utf8]{inputenc}
\usepackage{amsmath}
\usepackage{amssymb}
\usepackage{enumerate}
\usepackage{dsfont}
\usepackage{thmtools}
\usepackage{mathtools}
\usepackage{thm-restate}
\usepackage{multirow}
\usepackage{xcolor}
\usepackage{tikz}
\usetikzlibrary{shapes}
\usepackage{todonotes}
\usepackage{comment}
\usepackage{placeins}
\usepackage{nicefrac}
\usepackage{tabularx}
\usepackage{algorithm2e}
\usepackage{etoolbox}
\usepackage{thm-restate}

\hypersetup{colorlinks = true, citecolor=blue}

\usepackage{environ}
\makeatletter
\newsavebox{\measure@tikzpicture}
\NewEnviron{scaletikzpicturetowidth}[1]{%
  \def\tikz@width{#1}%
  \begin{lrbox}{\measure@tikzpicture}%
  \BODY
  \end{lrbox}%
  \pgfmathparse{#1/\wd\measure@tikzpicture}%
  \BODY
}
\makeatother

\newcommand{\Oh}{\ensuremath{\mathcal{O}}}

\usepackage{etoolbox}

\newcommand{\appsymb}{$\star$}
\newcommand{\appref}[1]{{\hyperref[proof:#1]{\appsymb}}}

\crefname{observation}{Observation}{Observations}

\newcommand{\mT}{\ensuremath{\mathcal{T}}}

\DeclareMathOperator{\poly}{poly}

\DeclareMathOperator{\td}{td}

\DeclareMathOperator{\cw}{cw}
\DeclareMathOperator{\tw}{tw}

\DeclareMathOperator{\Wh}{W[1]}

\newcommand{\nPHC}{{\normalfont coNP $\not \subseteq$ NP/poly}}

\newcommand{\problemdef}[3]{
    \begin{quote}
      \normalsize\textsc{#1} \smallskip \\
      \begin{tabularx}{0.9\textwidth}{@{}l@{\hspace{3pt}}X}
        \normalsize\textbf{Input:}    & \normalsize#2 \\
        \normalsize\textbf{Question:} & \normalsize#3
      \end{tabularx}
    \end{quote}
}

{\upshape\itshape}{\upshape\rmfamily}

\newcommand{\HC}{\textsc{Hamiltonian Cycle}\xspace}
\newcommand{\HP}{\textsc{Hamiltonian Path}\xspace}
\newcommand{\stueL}{\textsc{Unweighted Flexible Graph Connectivity}\xspace}
\newcommand{\stue}{\textsc{UFGC}\xspace}

\tikzstyle{para}=[rectangle,draw=black,minimum height=.8cm,fill=gray!10,rounded corners=1mm, on grid]

\newcommand{\tworows}[2]{\begin{tabular}{c}{#1}\\{#2}\end{tabular}}

\newcommand{\distto}[1]{\tworows{Distance to}{#1}}

\DeclareRobustCommand{\tikzdot}[1]{\tikz[baseline=-0.6ex]{\node[draw,fill=#1,inner sep=2pt,circle] at (0,0) {};}}

\definecolor{r}{rgb}{1.0, 0.4, 0.4}
\definecolor{r0}{rgb}{1.0, 0.7, 0.4}
\definecolor{r1}{rgb}{1.0, 0.4, 0.0}
\definecolor{r2}{rgb}{0.8, 0, 0.4}
\definecolor{r3}{rgb}{1, 0.4, 0.3}
\definecolor{b}{rgb}{0.4, 0.4, 1.0}
\definecolor{b0}{rgb}{0.4, 0.7, 1.0}
\definecolor{b1}{rgb}{0.0, 0.4, 1.0}
\definecolor{b2}{rgb}{0.4, 0.0, 0.8}
\definecolor{g}{rgb}{0.2, 0.9, 0.3}
\definecolor{y2}{rgb}{1, 1, 0.3}

\newcommand{\SSS}{\ensuremath{T}}

\DeclareMathOperator{\DP}{DP}

\DeclareMathOperator{\move}{move}
\newcommand{\movev}{\ensuremath{\move\text{-}v}}

\newcommand{\kommentar}[1]{}

\title{Finding a Sparse Connected Spanning Subgraph in a non-Uniform Failure Model} 

\titlerunning{Finding a Sparse Connected Spanning Subgraph in a non-Uniform Failure Model} 

\author{Matthias Bentert}{University of Bergen, Department of Informatics, Norway}{matthias.bentert@uib.no}{}{Supported by the  European Research Council (ERC) under the European Union’s Horizon 2020 research and innovation programme (grant agreement No. 819416).}

\author{Jannik Schestag}{Friedrich-Schiller-Universität Jena, Fakultät für Mathematik und Informatik,  Jena, Germany}{tschestag@tudelft.nl}{https://orcid.org/0000-0001-7767-2970}{Supported by the German Academic Exchange Service (DAAD), project 57556279.}

\author{Frank Sommer}{Friedrich-Schiller-Universität Jena, Fakultät für Mathematik und Informatik,  Jena, Germany}{frank.sommer@uni-jena.de}{https://orcid.org/0000-0003-4034-525X}{Supported by the Deutsche Forschungsgemeinschaft (DFG), project EAGR,\\ {KO~3669/{6-1}}.}

\authorrunning{M.~Bentert, S.~Schestag, and F.~Sommer}

\acknowledgements{This work was initiated at the research retreat of the Algorithmics and Computational Complexity group of TU Berlin held in Darlingerode in September 2022.}

\ccsdesc[500]{Theory of computation~Parameterized complexity and exact algorithms}
\ccsdesc[500]{Theory of computation~Graph algorithms analysis}

\keywords{Flexible graph connectivity, NP-hard problem, parameterized complexity, below-guarantee parameterization, treewidth}

\Copyright{Anonymous Author(s)}

\begin{document}

\maketitle

\begin{abstract}

We study a generalization of the classic \textsc{Spanning Tree} problem that allows for a non-uniform failure
model. 
More precisely, edges are either \emph{safe} or \emph{unsafe} and we assume that failures only affect unsafe edges. 
In \textsc{Unweighted Flexible Graph Connectivity} we are given an undirected graph~$G = (V,E)$ in
which the edge set~$E$ is partitioned into a set~$S$ of safe edges and a set~$U$ of unsafe edges and the task
is to find a set~$T$ of at most~$k$ edges such that~$T - \{u\}$ is connected and spans~$V$ for any unsafe edge~$u \in T$.
\textsc{Unweighted Flexible Graph Connectivity} generalizes both \textsc{Spanning Tree} and \textsc{Hamiltonian Cycle}.
We study \textsc{Unweighted Flexible Graph Connectivity} in terms of fixed-parameter tractability (FPT).
We show an almost complete dichotomy on which parameters lead to fixed-parameter tractability
and which lead to hardness. 
To this end, we obtain FPT-time algorithms with respect to the vertex deletion distance to cluster graphs and with respect to the treewidth. 
By exploiting the close relationship to \textsc{Hamiltonian Cycle}, we show that FPT-time algorithms for many smaller parameters are excluded under standard assumptions in parameterized complexity.
Regarding problem-specific parameters, we observe that \textsc{Unweighted Flexible Graph Connectivity} admits an FPT-time algorithm when parameterized by the number of unsafe edges. 
Furthermore, we investigate a below-upper-bound parameter for the number of edges of a solution. We show that this
parameter also leads to an FPT-time algorithm.
\end{abstract}

 \section{Introduction}
Computing a spanning tree is a fundamental task in computer science with a huge variety of applications in network design~\cite{FLL03} and clustering problems~\cite{XOX02}.
In \textsc{Spanning Tree}, one is given a graph~$G$, and the aim is to find a set~$T\subseteq E(G)$ of minimal size such that each pair of vertices in~$G$ is connected via edges in~$T$.
It is well known that \textsc{Spanning Tree} can be solved in polynomial time~\cite{K56,P57}.
This classic spanning-tree model has, however, a major drawback: all edges are seen as equal.
In many scenarios, for example in the construction of supply chains~\cite{LSDC06}, this is not sufficient: some connections might be more fragile than~others.
To overcome this issue, several different network-design and connectivity problems are studied with additional robustness constraints~\cite{FML22,SV14}.
In this work, we continue this line of research and investigate a graph model in which the edge set is partitioned into \emph{safe} edges~$S$ and \emph{unsafe} edges~$U$~\cite{AHM22,AHMS22,BCHI21}.
In contrast to several other variants of robust connectivity, these two edge types enable us to model a non-uniform failure scenario~\cite{AHM22}.
With these types at hand, we can relax the model of a spanning tree in the sense that one unsafe edge may~fail.
An edge set~$T$ of a graph~$G=(V,S,U)$ is a \emph{safe spanning connected subgraph} if~$T-\{e\}$ is a connected spanning subgraph for~$V$ for each unsafe edge $e \in T$~\cite{AHM22}.
This leads to the following problem:

\problemdef{\stueL}{A graph~$G=(V,S,U)$ and an integer~$k$.}{Is there a safe spanning connected subgraph~$T$ of~$G$ with~$|T| \leq k$?}

In the following, we refer to~$T$ as a solution. 
\stue{} generalizes several well-studied classic graph problems.
Examples include \textsc{$2$-Edge Connected Spanning Subgraph}~\cite{Huh04,SV14} ($S=\emptyset$) and \HC{}~\cite{ANS80,Golumbic04} ($S=\emptyset$ and~$k=|V(G)|$).
Thus, in sharp contrast to the classic \textsc{Spanning tree} problem, \stue{} is NP-hard.

\paragraph*{Related Work} 
If the graph is additionally equipped with an edge-cost function, the corresponding problem, in which one aims to find a solution of minimum total weight, is known as \textsc{Flexible Graph Connectivity}~\cite{AHM22,BCHI21}.
\textsc{Flexible Graph Connectivity} is mainly studied in terms of approximation:
Adjiashvili et al.~\cite{AHM22} provided a polynomial-time $2.527$-approximation algorithm which was improved by Boyd et al.~\cite{BCHI21} to a polynomial-time $2$-approximation.
Recently, a generalization called \textsc{$(p,q)$-Flexible Graph Connectivity} was introduced and studied in terms of approximation~\cite{BCHI21}.
In this model, up to $p$~unsafe edges may fail and the result is required to be $q$-edge connected.
Clearly, \textsc{Flexible Graph Connectivity} is the special case where~$p=1$ and~$q=1$.
Boyd et al.~\cite{BCHI21} provided a $(q+1)$-approximation for the case~$p=1$ and a $\Oh(q\log(n))$~approximation for the general case.
For the special case of~$p=2$, an improved approximation of~$\Oh(q)$ was provided by  Chekuri and Jain~\cite{CJ22}, and for the special case of~$q=1$ a constant-factor approximation was shown by Bansal et al.~\cite{BCGI22}.


Our model with safe and unsafe edges is based on prior work by e.g.\;Adjiashvili et al.~\cite{AHMS22}, who studied the approximability of the classic \textsc{$(s,t)$-Path} and \textsc{$(s,t)$-Flow} problems in this model.
Even previous to this, some studies already indirectly investigated problems in this setting:
There have been some studies of classic graph-connectivity problems where one wishes for the stronger requirement of~$2$-edge connectivity.
In terms of our setting, this corresponds to the case that all edges are unsafe. 
One prominent example is the aforementioned \textsc{$2$-Edge Connected Spanning Subgraph}~\cite{Civril23,HVV19,SV14}.
There exists a $4/3$-approximation~\cite{Civril23,HVV19,SV14}, and a~$6/5$-approximation in cubic graphs~\cite{Civril23}.
Furthermore, for its generalization \textsc{$q$-Edge Connected Spanning Subgraph} approximation algorithms of ratio~${1+1/(2q)+\Oh(1/(q^2))}$ are known~\cite{GG12}.
The problem \textsc{$q$-Edge Connected Spanning Subgraph} is also studied in terms of parameterized complexity:
Basavaraju et al.~\cite{BMRS17} provided an FPT-time algorithm with respect to the number of edges deleted from~$G$ to obtain a solution.
A variant of \textsc{$2$-Edge Connected Spanning Subgraph} in digraphs was studied by Bang-Jensen and Yeo~\cite{BY08}.
They provided an FPT-time algorithm for a solution-size related parameter below an upper bound.


\paragraph*{Our Results}
We study the parameterized complexity of \stue with respect to many natural parameters.
We investigate problem-specific as well as structural graph parameters.
A first idea to solve \stue might be the following: merge each safe component into a single vertex by computing an arbitrary spanning tree of this component and then use existing algorithms for \textsc{$2$-Edge Connected Spanning Subgraph} to solve the remaining instance.
This approach, however, does not yield a minimal solution: 
Consider a~$C_4$ on unsafe edges where we add one safe edge. 
Then, the unique minimal solution consists of all four unsafe edges.
Hence, we need new techniques to solve \stue optimally.


In terms of problem-specific parameters, we first study the parameterization by the number~$|U|$ of unsafe edges. 
We provide an FPT-time algorithm for~$|U|$ (\Cref{prop-fpt-u}), by exploiting the fact that \textsc{Disjoint Subgraphs} admits an FPT-time algorithm.
Second, we investigate parameterizations above a lower bound and below an upper bound for the solution size.
For the former parameterization, we obtain hardness due to the connection to \HC{} and for the latter we obtain an FPT-time algorithm (\Cref{thm-fpt-for-2n-ell}).
Our algorithm is inspired by an algorithm of Bang-Jensen and Yeo~\cite{BY08} for \textsc{Minimum Spanning Strong Subgraph}.
In this problem one is given a digraph~$D$ and one wants to find a minimum spanning strong digraph of~$D$.
The main technical hurdle in our adaption is that in \stue we have two different edge types which have to be treated differently compared with the problem studied by Bang-Jensen and Yeo~\cite{BY08}.


\begin{figure}[t]
\centering

\begin{tikzpicture}[node distance=2*0.45cm and 3.7*0.38cm, every node/.style={scale=0.57}]
\linespread{1}
\node[para,fill=g] (vc) {Minimum Vertex Cover};
\node[para, fill=g, xshift=4cm] (ml) [right=of vc] {Max Leaf \#};
\node[para, xshift=-3.4cm,fill=g] (dc) [left=of vc] {\distto{Clique}};

\node[para, yshift=-5mm,xshift=-14mm] (mcc) [below left=of dc] {\tworows{Minimum}{Clique Cover}}
edge (dc);
\node[para,yshift=-5mm,xshift=-12.5mm] (dcc) [below= of dc] {\distto{Co-Cluster}}
edge (dc)
edge[bend left=10] (vc);
\node[para,fill=g,xshift=13mm,yshift=-5mm] (dcl) [below= of dc] {\distto{Cluster}}
edge (dc)
edge (vc);
\node[para, xshift=13.5mm,yshift=-5mm,fill=g] (ddp) [below=of vc] {\distto{disjoint Paths}}
edge (vc)
edge (ml);
\node[para,yshift=-5mm,fill=g] (fes) [below =of ml] {\tworows{Feedback}{Edge Set}}
edge (ml);
\node[para,yshift=-5mm,fill=g] (bw) [below right=of ml] {Bandwidth}
edge (ml);
\node[para,xshift=-18mm,yshift=-5mm,below=of vc] (nd) {\tworows{Neighborhood}{Diversity}}
edge (vc)
edge (dc);

\node[para,xshift=2mm] (is) [below=of mcc] {\tworows{Maximum}{Independent Set}}
edge (mcc);
\node[para,xshift=5mm] (dcg) [below= of dcc] {\distto{Cograph}}
edge (dcc)
edge (dcl);
\node[para,xshift=5mm] (dig) [below= of dcl] {\distto{Interval}}
edge (dcl)
edge (ddp);
\node[para,fill=g,yshift=-10mm] (fvs) [below= of ddp] {\tworows{Feedback}{Vertex Set}}
edge (ddp)
edge (fes);
\node[para, xshift=3.5mm, yshift=-0mm,fill=g] (td) [right=of ddp] {Treedepth}
edge [bend right=28] (vc);
\node[para,fill=r3] (mxd) [below= of bw] {\tworows{Maximum}{Degree}}
edge (bw);

\node[para, fill=r3] (ds) [below=of is] {\tworows{Minimum}{Dominating Set}}
edge (is);
\node[para, yshift=-55mm, xshift=-40mm,fill=r3] (dbp) [below left= of fvs] {\distto{Bipartite}}
edge (fvs);
\node[para, xshift=0mm, yshift=-7mm,fill=g] (dop) [below= of fvs] {\distto{Outerplanar}}
edge (fvs);
\node[para, yshift=-13mm,fill=g] (pw) [below= of td] {Pathwidth}
edge (ddp)
edge (td)
edge (bw);
\node[para,fill=r3,yshift=-10mm] (hid) [below= of mxd] {$h$-index}
edge (ddp)
edge (mxd);
\node[para, fill=r3] (gen) [left= of hid] {Genus}
edge (fes);

\node[para,fill=r3,xshift=1.5mm] (mxdia) [below=of ds] {\tworows{Max Diameter}{of Components}}
edge (ds)
edge[bend right=5] (td)
edge[bend right=37] (nd);
\draw (mxdia.north east) to (dcg);
\node[para, yshift=5mm,fill=g] (tw) [below right= of dop] {Treewidth}
edge (dop)
edge (pw);
\node[para, yshift=-15mm,fill=r3] (acn) [below = of gen] {\tworows{Acyclic}{Chromatic \#}}
edge (hid)
edge (gen)
edge (tw);
\node[para, yshift=-2.5mm,fill=r3] (dpl) [below = of dop] {\distto{Planar}}
edge (dop)
edge (acn);
\node[para, xshift=-15mm, yshift=-02mm,fill=y2] (clw) [below= of dpl] {Clique-width}
edge (dcg)
edge[bend right=15] (tw);
\node[para,fill=r3] (avgdist) [below=of mxdia] {\tworows{Average}{Distance}}
edge (mxdia);
\node[para, xshift=10mm, yshift=0mm,fill=r3] (dch) [right= of avgdist] {\distto{Chordal}}
edge (dig)
edge (fvs);
\node[para,fill=r3] (deg) [below= of acn] {Degeneracy}
edge (acn);
\node[para, yshift=-12mm, xshift=11mm,fill=r3] (box) [left=of clw] {Boxicity}
edge (dig)
edge (nd)
edge (acn);
\node[para,fill=r3] (cn) [below =of dbp] {Chromatic \#}
edge (deg)
edge (dbp);
\node[para, yshift=0mm, xshift=-4mm,fill=r3] (dpf) [left= of cn] {\distto{Perfect}}
edge (dch)
edge (dcg)
edge (dbp);
\node[para,fill=r3] (avd) [below=of deg] {\tworows{Average}{Degree}}
edge (deg);

\node[para,fill=r3] (mnd) [below=of avd] {\tworows{Minimum}{Degree}}
edge (avd);

\node[para,fill=r3] (mc) [below=of cn] {\tworows{Maximum}{Clique}}
edge (cn);
\node[para, yshift=5mm, xshift= 15mm,fill=r3] (cho) [right= of mc] {Chordality}
edge (box)
edge (dch)
edge (cn)
edge (dcg);
\node[para,fill=r3] (gi) [below=of avgdist] {Girth}
edge (avgdist);
\end{tikzpicture}
\caption{
The relations between structural graph parameters and our respective results for \stue.
A parameter~$k$ is marked
green (\tikzdot{g}) if \stue{} admits an FPT-time algorithm for~$k$,
yellow (\tikzdot{y2}) if it is W[1]-hard with respect to~$k$,
and red (\tikzdot{r3}) if it is NP-hard for constant~$k$ (para-NP-hard).
We do not know the status for parameters with white boxes.
An edge from a parameter~$\alpha$ to a parameter~$\beta$ below~$\alpha$ means that there is a function~$f$ such that~$\beta \leq f(\alpha)$ in every graph.
Hardness results for~$\alpha$ imply the same hardness results for~$\beta$ and an FPT-time algorithm for~$\beta$ implies an FPT-time algorithm for~$\alpha$.
}
	\label{fig-results}
\end{figure}
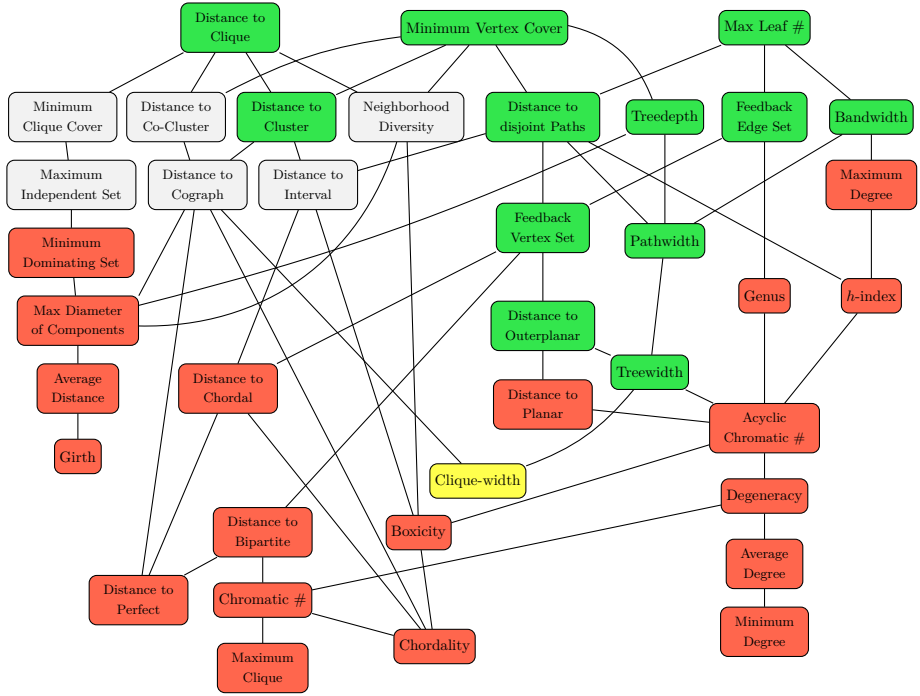

We also study parameterization by structural graph parameters; see \Cref{fig-results} towards a dichotomy for \stue.
For many parameters such as maximum degree and domination number, we obtain para-NP-hardness, due to the connection to \HC.
For other parameters, we obtain FPT-time algorithms. 
For example, for the treewidth~$\tw$ of the input graph and the (vertex) deletion distance to cluster graphs. 
With these two algorithms at hand, we obtain an almost complete border between parameters which allow for FPT-time algorithms and those that do not.

For the treewidth~$\tw$, we present an FPT-time algorithm based on dynamic programming with running time~$n \cdot 2^{\Oh(\tw \log(\tw))}$ (\Cref{prop-fpt-tw}).
In order to achieve this running time, we define and exploit an encoding of size~$\tw^{\tw}$ to check all possibilities on how the current bag of the tree decomposition is connected with the already considered vertices.
Finally, we show our main technical result:  \stue{} parameterized by the vertex-deletion distance to cluster graphs admits an FPT-time algorithm (\Cref{thm-fpt-dist-cluster}).
Therein, we use a combination of the following two main ingredients:
First, given a modulator~$K$, that is, a set of vertices such that~$G' = G - K$ is a cluster graph, we can safely connect all vertices in~$K$ using vertices from~$\Oh(|K|)$ cliques in~$G'$ in any solution.
We call these cliques the \emph{backbone} of the solution and we can guess in FPT time the structure of the backbone.
Second, we show how to compute the size of a smallest solution that implements a backbone in FPT time using algorithms for both \textsc{Maximum Bipartite Matching} and \textsc{Disjoint Subgraphs} in the process.

\section{Preliminaries}
\label{sec-prelims}

For~$n \in \mathds{N}$, we denote by~$[n]$ the set $\{ 1, \dots, n \}$.
Throughout this work, all logarithms have 2 as their base.

\paragraph*{Graph Notation}
For a graph~$G=(V,S,U)$, we denote by~$V(G)$ and~${E(G)\coloneqq S\cup U}$ its \emph{vertex set} and \emph{edge set}, respectively. 
Furthermore, we denote the number of vertices by~$n \coloneqq |V(G)|$.
Let~$Z \subseteq V(G)$ be a vertex set.
We denote the \emph{subgraph induced} by~$Z$ by~$G[Z]$ and the graph obtained by \emph{removing} the vertices of~$Z$ by~$G - Z \coloneqq G[V(G) \setminus Z]$.
We denote by~${N_G(Z)\coloneqq\{ y \in V(G) \setminus Z : \{y,z\} \in E(G), z \in Z \}}$ and~${N_G[Z]\coloneqq N_G(Z) \cup Z}$, the \emph{open} and \emph{closed neighborhood} of~$Z$, respectively.
By~$\Delta_G\coloneqq \max_{v\in V(G)}|N(v)|$ we denote the \emph{maximum degree} of~$G$.
For all these notations, whenever~$Z$ is a singleton~$\{ z \}$ we may write~$z$ instead of~$\{z\}$.
We may drop the subscript~$\cdot_G$ when it is clear from the context. 
In a \emph{cluster graph} each connected component is a clique.

A \emph{tree decomposition} of graph~$G$ is a pair~$(\mT=(\mathcal{V}, \mathcal{A}), \{X_t:t\in\mathcal{V}\})$ consisting of a directed tree~$\mT$ with root~$r\in \mathcal{V}$ and \emph{bags}~$X_t\subseteq V(G)$ such that
\begin{enumerate}
\item for each~$v\in V(G)$, there is at least one~$t\in \mathcal{V}$ with~$v\in X_t$,
\item for each~$\{u,w\}\in E(G)$, there is at least one~$t\in \mathcal{V}$ such that~$u\in X_t$ and~$w\in X_t$, and
\item for each~$v\in V(G)$, the subgraph~$\mT[\mathcal{V}_v]$ is connected, where ${\mathcal{V}_v \coloneqq \{t\in \mathcal{V}\mid v\in X_t\}}$.
\end{enumerate}

By~$Y_t\coloneqq(\bigcup_{t' \text{ is a descendant of } t} X_{t'})\setminus X_t$ we denote the \emph{forgotten vertices} in note~$t\in\mathcal{V}$. 
The~\emph{width of a tree decomposition} is the size of the largest bag minus one and the \emph{treewidth}~$\tw(G)$ of a graph~$G$ is the minimal width of any tree decomposition of~$G$.

We consider tree decompositions with specific properties.
A node~$t\in \mathcal{V}$ is called:
\begin{enumerate}
\item a \emph{leaf node} if~$t$ has no child nodes in~$\mT$,
\item a \emph{forget node} if~$t$ has exactly one child node~$t'$ in~$\mT$ and~$X_{t'} = X_t \cup \{v\}$ for some~$v\in V(G)\setminus X_t$,
\item an \emph{introduce node} if~$t$ has exactly one child node~$t'$ in~$\mT$ and~$X_{t'} = X_t \setminus \{v\}$ for some~$v\in V(G)\setminus X_{t'}$, or
\item a \emph{join node} if~$t$ has exactly two child nodes~$t_1$ and~$t_2$ in~$\mT$ and~$X_t = X_{t_1} = X_{t_2}$.
\end{enumerate}

A tree decomposition is called~\emph{nice} if every node~$t\in \mathcal{V}$ is either a leaf node, a forget node, an introduce node, or a join node.
Given a tree decomposition of width~$\omega$, in polynomial time one can compute a nice tree decomposition also of width~$\omega$~\cite{Kloks94}.
For a graph~$G$, a nice tree decomposition of width at most~$2\tw(G)+1$ can be computed in $\Oh(2^{\Oh(\tw(G))} \cdot n)$~time~\cite{Kor22}.

Let~$u,v\in V(G)$.
We say that~$u$ and~$v$ are \emph{safely connected} if there exists a path from~$u$ to~$v$ using only safe edges or if there exist two paths~$P_1$ and~$P_2$ from~$u$ to~$v$ such that~$E(P_1)\cap E(P_2)\subseteq S$.
We say that~$u$ and~$v$ are \emph{unsafely connected} if~$u$ and~$v$ are not safely connected and if there exists a path from~$u$ to~$v$.

For more details on graph notation, we refer to the textbook by West~\cite{West00}.

\paragraph*{Parameterized Algorithmics}
A parameterized problem~$L$ is a set of instances of the form~$(I, k)$ where~$I\in\Sigma^*$ for a finite
alphabet~$\Sigma$, and~$k \in \mathds{N}_0$ is the \emph{parameter}.
A parameterized problem is \emph{fixed-parameter tractable} if every instance~$(I, k)$ can be solved in $f(k) \cdot |I|^{\Oh(1)}$~time for some computable function~$f$.
An algorithm with such a running time is referred to as \emph{FPT-time algorithm}.
A \emph{parameterized reduction} is an algorithm that takes an instance~$I_1 = (x_1, k_1)$ of a parameterized problem~$L_1$ and transforms~$I_1$ into an instance~$I_2 = (x_2, k_2)$ of a parameterized problem~$L_2$ such
that~$I_1 \in L_1$ if and only if~$I_2 \in L_2$ and~$k_2 \le g(k_1)$ with running time~$f(k_1) \cdot |x_1|^{\Oh(1)}$
for some computable functions~$f$ and~$g$. 
Parameterized reductions can be used to show that a parameterized problem~$L$ is not likely to be fixed-parameter tractable by reducing some \emph{W[1]-hard} parameterized problem to~$L$.
An algorithm with running time~$f(k)\cdot |I|^{g(k)}$ for computable function~$f$ and~$g$ is called \emph{XP-algorithm}.
A problem is \emph{para-NP-hard} for parameter~$k$ if it is NP-hard even for constant values of~$k$.
The \emph{Exponential Time Hypothesis (ETH)} is a standard complexity theoretical conjecture used to prove lower bounds. 
The ETH implies that \textsc{3-SAT} cannot be solved in $2^{o(|\varphi|)}$~time where~$\varphi$ denotes the input formula~\cite{IPZ01}.

For more details on parameterized complexity, we refer to the textbooks by Cygan et al.~\cite{CFK+15} and by Downey and Fellows~\cite{DF13}.

\section{Basic Observations}
\label{sec-basic-obs}
In this section, we explore the aforementioned connection between \stue{} and \HC{} and some implications thereof.
To this end, note that if the graph contains no safe edges, then each vertex needs to be contained in at least one cycle in the solution.
We show that a solution of size~$n$ to \stue{} (if such a solution exists) corresponds to a Hamiltonian cycle.


\begin{observation}
\label{lem-correspondin-equivalent}
Let~$G$ be a graph and let~${I \coloneqq (G'=(V(G),\emptyset,E(G)),n)}$.
Then~$G$ contains a Hamiltonian cycle if and only if ~$I$ is a yes-instance of \stue.
\end{observation}

\begin{proof}
Note that since~$I$ does not contain any safe edges, each solution of~$I$ contains at least~$n$~edges.

$(\Rightarrow)$ let~$C$ be a Hamiltonian cycle of~$G$.
Then,~$C$ is a solution of size~$n$ for~$I$.

$(\Leftarrow)$ 
Let~$T\subseteq E(G)$ be a solution of size~$n$ for~$I$.
Since~$T$ is spanning, each edge is unsafe, and since~$|T|=n$, we conclude that each vertex in the edge-induced subgraph of~$T$ has degree exactly two.
Thus, $T$ is a disjoint union of cycles.
Since~$T$ is a solution to \stue{}, there is a path between each pair of vertices in~$T$, that is, the solution has exactly one connected component.
Thus, $T$ is a Hamiltonian cycle for~$G$.
\end{proof}

Since \Cref{lem-correspondin-equivalent} describes a polynomial-time reduction from \HC{} to \stue, we can directly transfer any NP-hardness results for \HC{} on restricted graph classes to \stue{} for the special case that there are no safe edges.
It is known that \HC{} remains NP-hard on split graphs~\cite{Golumbic04}, and on graphs with constant maximum degree~\cite{GJ79}.
More precisely, \HC{} remains NP-hard on subcubic bipartite planar graphs~\cite{ANS80}.
Moreover, \HC{} is NP-hard for graphs with exactly one universal vertex, as shown next.

\begin{observation}
\label{obs-dom-number-one}
\HC{} is NP-hard even if there is exactly one universal vertex.
\end{observation}

\begin{proof}
We reduce from the NP-hard \HP{} problem on graphs with constant maximum degree~\cite{GJ79}.
In \HP{} the aim is to find a path which visits every vertex exactly once.
Clearly, since the \HP{} instance has a constant maximum degree, we can assume that there is no universal vertex.

Let~$G$ be an instance of \HP{}. 
The instance~$G'$ of \HC{} consists of a copy of~$G$ plus a new universal vertex~$u$.
Since~$G$ has no universal vertex, we obtain that~$G'$ has exactly one universal vertex.
To see that both instances are equivalent, note that~$P$ is a Hamiltonian path of~$G$ if and only if the graph obtained by connecting both endpoints of~$P$ with the new vertex~$u$ is a Hamiltonian cycle of~$G'$.
\end{proof}

We obtain the following simple corollary for \stue.
\begin{corollary}
\label{cor-ust-np-h}
\stue{} is NP-hard, even if there are no safe edges, $k = n$, and the graph $G$~$a)$ is subcubic, planar and bipartite,~$b)$ is a split graph, or~$c)$ has domination number one.
\end{corollary}

Moreover, since \HC{} is W[1]-hard parameterized by clique-width~$\cw$~\cite{FGLS10}, and cannot be solved in $f(\cw)\cdot n^{o(\cw)}$~time assuming the ETH~\cite{FGLSZ19}, we obtain the following.

\begin{corollary}
\stue{} is $\Wh$-hard with respect to the clique-width~$\cw$, even if~$S=\emptyset$ and~$k=n$. Moreover, this restricted version cannot be solved in~$f(\cw)\cdot n^{o(\cw)}$~time unless the ETH fails.
\end{corollary}

We conclude this section with a proof that \stue{} is NP-hard even if the bisection width of the input graph is one.

\begin{observation}
\label{obs-bisection-1}
\stue{} is NP-hard even if the bisection width is one.
\end{observation}

\begin{proof}
Let~$(G=(V,S,U),k)$ be an instance of \stue.
We construct an equivalent instance~$(G',2k+1)$ of \stue{} with bisection width one as follows.
The graph consists of two copies of~$G$ and one safe edge between the two copies (the endpoints of this edge can be chosen arbitrarily).
Observe that~$G'$ has bisection width one.

If~$G$ contains a safe spanning connected subgraph of size~$k$, then we find a solution for~$G'$ of size~$2k+1$ by taking the same solution in both copies and adding the safe edge between these two solutions.
Conversely, if~$G'$ contains a safe spanning connected subgraph of size~$2k+1$, then it needs to contain the safe edge between the two copies of~$G$.
By the pigeonhole principle, there is a safe spanning connected subgraph of size at most~$k$ in one of the two copies.
This is by definition a solution of size at most~$k$ for the input graph~$G$.
\end{proof}

\section{Problem-Specific Parameters}
\label{sec-prob-related}

In this section, we study problem-specific parameters.
In particular, we show that \stue{} is fixed-parameter tractable when parameterized by the number of unsafe edges in the input graph and we study some \emph{above-lower-bound} and \emph{below-upper-bound} parameterizations for the solution size.
Such parameters are frequently studied~\cite{MR99,MRS09}.
However, since \Cref{lem-correspondin-equivalent} shows hardness for \stue{} where~$k=n$ and since each solution has at least $n-1$~edges, the parameterization above the lower bound of~$n-1$ is hopeless.
Hence, we focus on a parameterization below an upper bound.
More precisely, we show that an optimal solution of~$G$ consists of at most $2n-4$~edges.
Then, we aim to find a solution for~$G$ with $k=2n-4-\ell$~edges for small values of~$\ell$.



We start with the number~$|U|$ of unsafe edges.
To this end, we start with a simple observation.
\begin{observation}
	\label{obs:cycle}
	Let~$(G=(V, S, U),k)$ be an instance of \stue.
	If a solution contains a cycle consisting only of safe edges, then we can remove any edge in the cycle to get another solution.
\end{observation}

\begin{proposition}
\label{prop-fpt-u}
	\stue{}  is fixed-parameter tractable when parameterized by $|U|$.
\end{proposition}

\begin{proof}
We set~$\ell\coloneqq |U|$.
We start by guessing\footnote{Whenever we pretend to guess something, we iterate over all possibilities. We then check whether the current iteration leads to a solution. If it does, then we say that the guess was correct. Note that there is a solution if and only if there is a correct guess.} the set~$K \subseteq U$ of unsafe edges in some optimal solution~$\SSS$.
This can be done in~$\Oh(2^\ell)$ time.
Let~$F$ be the set of vertices that are an endpoint of an edge in~$K$, that is, $F = \{v \mid \exists\,e \in K.\ v \in e\}$.
Note that~$|F| \leq 2|K| \leq 2\ell$.
Moreover, let~$J = \SSS \setminus K$ be the set of safe edges in the solution~\SSS{} and let~$G_S = (V,S)$ and~$G' = (V, J)$ be the graphs that contain all vertices and a) all safe edges and b) the edges in~$J$, respectively.
By \cref{obs:cycle}, we may assume that~$G'$ is a forest.
Note that since~$G'$ is a forest, its number~$|J|$ of edges is~$n - c$ where~$c$ is the number of connected components in~$G'$.
Moreover, each connected component in~$G'$ has to contain at least one vertex in~$F$ (assuming~$F \neq \emptyset$)\footnote{Note that if~$F = \emptyset$, then~$K = \emptyset$ and \stue{} boils down to the classic polynomial-time solvable \textsc{Spanning Tree} problem.} as otherwise any vertex in the connected component is not connected to any vertex in~$F$ in~$G_S$, a contradiction to~$\SSS$ being a solution.
Hence, we next guess which vertices in~$F$ belong to the same connected components in~$G'$, that is, we guess a partition of the vertices in~$F$.
This takes~$B(|F|) \leq B(2\ell)$ time where~$B(i)$ is the~$i$\textsuperscript{th} Bell number.
It remains to check the following three points:
\begin{enumerate}
	\item Is there a way to connect all vertices in~$F$ that are guessed to be in the same connected component in~$G'$ such that these connected components are pairwise distinct?
	\item Can we connect all remaining vertices to at least one of these connected components?
	\item And is the resulting graph actually a solution of \stue?
\end{enumerate} 

Note that we can check the second and third point independent of the first point as the second point is equivalent to the question ``Is~$(V,S \cup K)$ a connected graph?'' and the third point is true if and only if each pair of vertices in~$F$ is safely connected.
The latter is completely determined by~$K$ and which vertices in~$F$ belong to the same connected component in~$G'$, that is, the guessed partition of~$F$.
It remains to check whether there is a set of trees in~$G_S$ to connect all vertices in~$F$ that are guessed to be in the same connected component in~$G'$ such that these connected components are pairwise distinct.
This problem is also known as \textsc{Disjoint Connected Subgraphs}. 
Robertson and Seymour~\cite{RS95} showed that there exists a computable function\footnote{Unfortunately, not much is known about the function~$f$ other than that it is computable.}~$f$ such that \textsc{Disjoint Connected Subgraphs} can be solved in~$f(k') \cdot n^3$ time, where~$k'$ is the number of terminals to be connected.
We can use this algorithm to check the first point in~$f(2\ell) \cdot n^3$ time.
Note that the overall running time is~$\Oh(2^\ell \cdot B(2\ell) \cdot f(2\ell) \cdot n^3)$, which is fixed-parameter tractable for the parameter~$\ell$.
\end{proof}

We continue by presenting an FPT-time algorithm for the below-upper-bound parameter~$\ell \coloneqq 2n - 4 - k$ for \stue{}.
We start with some additional required notation.
A vertex~$v$ is called a \emph{cut-vertex} if~$G-v$ has more connected components than~$G$.
Graphs without cut-vertex are called \emph{$2$-vertex connected}.
An edge~$e$ is a \emph{bridge} if~$G\setminus\{e\}$ has more connected components than~$G$.
Observe that if there is a bridge in~$G$ which is unsafe, then there does not exist any safe spanning connected subgraph of~$G$.
Hence, in the following, we assume that each bridge of~$G$ corresponds to a safe edge. 
A graph without bridge is called \emph{$2$-edge connected}.

As previously discussed, our algorithm is inspired by an algorithm of Bang-Jensen and Yeo~\cite{BY08} for \textsc{Minimum Spanning Strong Subgraph}.
Recall that in this problem, one is given a digraph~$D$ and one wants to find a minimum spanning strong digraph of~$D$.
Each solution for this problem consists of at most $2n-2$~arcs.
Bang-Jensen and Yeo~\cite{BY08} provided an FPT-time algorithm for the difference~$\ell$ if one searches for a solution with $2n-2-\ell$~arcs.
An important tool for the FPT-time algorithm are ear-decompositions~\cite{West00}.
Note that a graph~$G$ has an ear-decomposition if and only if~$G$ is $2$-edge connected~\cite{R39}.

\begin{definition}[\cite{West00}]
\label{def-ear-decomp}
An \emph{ear decomposition} of an undirected $2$-edge connected graph~$G$ is a sequence~$\mathcal{E}\coloneqq (P_0, P_1, \ldots, P_t)$, where~$P_0$ is a cycle, and each~$P_i$ for~$i\ge 1$ is a path or a cycle fulfilling the following properties:
\begin{enumerate}
\item $P_i$ and~$P_j$ are edge-disjoint for~$i\ne j \in[0,t]$.

\item For~$i\in[t]$ we denote by~$G_i$ the graph with vertex set~$V(G_i)\coloneqq \bigcup_{j \in [0,i]} V(P_j)$ and edge set~$E(G_i)\coloneqq  \bigcup_{j \in [0,i]} E(P_j)$.
If~$P_i$ is a cycle, then~$P_i$ has exactly one common vertex with~$V(G_{i-1})$, and if~$P_i$ is a path, then the intersection of~$V(P_i)$ and~$V(G_{i-1})$ consists of both endpoints of~$P_i$.

\item $G_t = G$.
\end{enumerate}
\end{definition}

Each~$P_i$ is denoted as an \emph{ear}.
The \emph{size} of an ear is the number of edges of the ear.
If the size of~$P_i$ is one, then we say that~$P_i$ is \emph{trivial}.
Otherwise, we say that~$P_i$ is \emph{non-trivial}.
Observe that we may assume that in any ear decomposition all non-trivial ears appear before any trivial ears.
Recall that we assume that each bridge of~$G$ is a safe edge.
In the following, we only consider graphs that are $2$-edge connected.
We will later generalize the result to graphs that are not $2$-edge connected.
Observe that each cycle of a $2$-edge connected graph can serve as ear~$P_0$ in some ear decomposition.

The following lemma follows directly by \Cref{def-ear-decomp}.

\begin{lemma} 
\label{lem-ear-decomp-triv-sol}
Let~$G$ be a $2$-edge connected graph and let~$\mathcal{E}=(P_0,\ldots, P_q,P_{q+1},\ldots, P_t)$ be an ear decomposition of~$G$ such that~$P_{q+1}, \ldots, P_t$ are the trivial ears of~$\mathcal{E}$.
Then, the edges of~$G$ corresponding to~$T=E(P_0)\cup\ldots\cup E(P_q)$ are a safe spanning connected subgraph for~$G$.
\end{lemma}
\begin{proof}
We have to show that~$T-\{e\}$ is a connected spanning subgraph for each unsafe edge~$e$ in~$U$.
Since each vertex has at least two incident edges in~$T$, we only need to show that~$T - \{e\}$ is connected for each edge~$e \in T$.
We verify this statement inductively.
Clearly, the statement is true for~$P_0$.
In the following, we assume that the statement is shown for~$G_{i-1}$, the graph obtained from ears~$P_0$ to~$P_{i-1}$.
Next, we distinguish whether~$P_i$ is a cycle or a path.
If~$P_i$ is a cycle, then the only interesting vertex pairs are those where one endpoint~$u$ is in~$V(G_{i-1})$ and the other~$v$ is in~$P_i$.
Without loss of generality, assume that~$w=V(G_{i-1})\cap P_i$.
Then, since the statement is true for~$V(G_{i-1})$, the vertices~$u$ and~$w$ are safely connected.
Furthermore, since~$P_i$ is a cycle,~$v$ and~$w$ are safely connected via vertices in~$V(G_i)$.
Thus, $u$ and~$v$ are safely connected via vertices in~$P_i$.
If~$P_i$ is a path, then let~$u$ and~$v$ be the endpoints of~$P_i$.
By the induction hypothesis, in~$G_{i-1}$ each vertex is safely connected to~$v$.
Hence, the removal of any edge in~$E(G_{i-1})$ does not disconnect any vertex from~$v$ and this also holds in~$G_i$ as each vertex in~$P_i$ is connected to~$v$ via~$P_{i}$.
The removal of any edge in~$P_i$ does also not disconnect any vertex from~$v$ as all vertices that are no longer reachable from~$v$ via~$P_i$ are still connected to~$u$ and~$u$ is safely connected to~$v$.
This concludes the induction step and the entire proof.
\end{proof}

\begin{lemma} 
\label{lem-ear-decomp-lage-ears}
Let~$G$ be a $2$-edge connected graph, let~$(P_0, \ldots, P_t)$ be the non-trivial ears of an ear decomposition of~$G$, and let~$\mathcal{P}$ be the set of ears~$P_i$ with~$i\ge 1$ with~$V(P)\ge 3$.
Then, $G$ has a safe spanning connected subgraph with at most ${2n-(|V(P_0)|+|\mathcal{P}|)}$~edges.
\end{lemma}


\begin{proof}
By \Cref{lem-ear-decomp-triv-sol},~$(P_0, \ldots, P_t)$ is a safe spanning connected subgraph~$T$ for~$G$ with~$|E(P_0)|+\ldots +|E(P_t)|$~edges on $n$~vertices.
Observe that~$P_0$ contributes exactly~$|P_0|$~vertices and edges to~$T$.
Furthermore, each~$P_i$ for~$i\in[t]$ contributes exactly~$|P_i|-1$~vertices and $P_i$~edges to~$T$.
Hence,~$t=|E(T)|-|V(T)|$.

Since each ear in~$\mathcal{P}$ contributes at least two vertices to~$F$ end each remaining ear exactly one vertex, we conclude that~$n\ge |V(P_0)|+2|\mathcal{P}|+(t-|\mathcal{P}|)$.
Hence, we obtain that~$|E(T)|=|V(T)|+t\le n + (n-|V(P_0)|-|\mathcal{P}|)=2n-(|V(P_0)|+|\mathcal{P}|)$.
\end{proof}

Next, we use \Cref{lem-ear-decomp-lage-ears} to provide an upper bound for the size of any safe spanning connected subgraph for \stue.

\begin{proposition} 
\label{ref-ub-solution-ust}
Each $2$-edge connected graph~$G$ with $n\ge 4$~vertices has a safe spanning connected subgraph with at most $2n-4$~edges.
\end{proposition}
\begin{proof}
Let~$\mathcal{E}=(P_0,\ldots, P_t)$ be an ear decomposition of~$G$  such that~$P_0$ is a longest cycle of~$G$.
\Cref{lem-ear-decomp-lage-ears} implies that if~$|V(P_0)|\ge 4$, there exists a safe spanning connected subgraph with at most $2n-4$~edges.
Hence, in the following, we assume that each cycle in~$G$ has length~$3$; and thus in particular~$P_0$ is a triangle.
If there is no~$P_1$, then~$G$ has only $3$~vertices, a contradiction to our assumption that~$n\ge 4$.
Thus, in the following, we assume that there exists some~$P_1$.
Next, we distinguish whether~$P_1$ is a cycle or a path. 

First assume that~$P_1$ is a path. 
Then,~$G[P_0\cup P_1]$ has a cycle of length at least~$4$, a contradiction to our assumption that each cycle in~$G$ has length~$3$.

Second, assume that~$P_1$ is a cycle.
Since each cycle in~$G$ is a triangle, also~$P_1$ is a triangle.
Note that~$G[P_0\cup P_1]$ has $5$~vertices and $6$~edges.
Since each non-trivial ear with $r$~vertices increases the number of edges by exactly~$r+1$, we obtain a safe spanning connected subgraph with at most $2n-4$~edges.
\end{proof}

We next show that in yes-instances, the size of each $2$-vertex connected component can be bounded by a function in~$\ell$.

\begin{lemma}
\label{lem-2-vertex-connected-component}
Let~$G$ be a $2$-vertex connected graph with at least $4$~vertices.
In polynomial time, we can either correctly decide whether~$(G,k)$ is a yes-instance or a no-instance, or compute an equivalent instance~$(G',k')$ with $8\ell^2+31\ell+28$~vertices.
\end{lemma}
\begin{proof}
Since~$G$ has at least $4$~vertices and since~$G$ is $2$-vertex connected,~$G$ contains a cycle~$C$ of length at least~$3$.
Let~$P_0\coloneqq C$ be the initial cycle of an ear decomposition of~$G$.
Next, we add to~$P_0$ a maximal sequence~$\mathcal{P}=(P_1,\ldots ,P_s)$ of ears of size at least~$3$.
We denote the vertex set of all these ears by~$X=V(P_0)\cup V(P_1)\cup \ldots \cup V(P_s)$.
Next, we show the following property of~$Y\coloneqq V(G)\setminus X$.

\begin{claim}
\label{claim-y-is}
$Y$ is an independent set.
\end{claim}
\begin{claimproof}
Assume towards a contradiction that there is an edge~$\{v,w\}$ where both endpoints are contained in~$Y$.
Let~$Y'$ be the set of vertices in~$Y$ which have degree at least one in~$G[Y]$.
Observe that if no vertex in~$Y'$ has a neighbor in~$X$, then~$G$ is not connected, a contradiction.
Thus, there exists at least one vertex~$y\in Y'$ which has a neighbor~$x\in X$.
Since~$y\in Y'$, there exists at least one neighbor~$w$ of~$y$ in~$Y'$.

Since~$G$ is 2-vertex connected, there is a path~$P$ from~$w$ to~$x$ that does not include~$y$.
Let~$z$ be the first vertex in~$V(P) \cap X$.
Then, the edges~$\{x,y\}$ and~$\{y,w\}$ plus the subpath of~$P$ from~$w$ to~$z$ form an ear of size at least~$3$, a contradiction to our assumption that we added a maximal sequence of ears of size at least~$3$.
\end{claimproof}

Next, we show that we can assume that~$|X|\le 2\ell +4$ since otherwise the instance is trivial.
According to \Cref{lem-ear-decomp-lage-ears}, we know that~$G$ has a safe spanning connected subgraph with at most $2n-(|V_{P_0}|+|\mathcal{P}|)=2n-4-(|V_{P_0}|+|\mathcal{P}|-4)$~edges.
If~$|V(P_0)|+|\mathcal{P}|-4\ge\ell$, then~$(G,k)$ is a yes-instance.
Hence, in the following, we assume that~$|V(P_0)|+|\mathcal{P}|-4<\ell$.
Furthermore, since~$|V(P_0)|\ge 3$, we obtain that~$|\mathcal{P}|<\ell+1$.
Since~$Y$ is an independent set by~\Cref{claim-y-is}, we obtain a safe spanning connected subgraph for~$G$ by adding an ear of size~$2$ to~$\mathcal{P}$ for each vertex in~$Y$.
Note that the resulting safe spanning connected subgraph has exactly $2|Y|+|X|+|\mathcal{P}|$~edges.
Since~$|X|+|Y|=n$, we obtain  that this safe spanning connected subgraph has $2n-|X|+|\mathcal{P}|=2n-4-(|X|-|\mathcal{P}|-4)$~edges.
Now, if~$|X|-|\mathcal{P}|-4\ge\ell$, we observe that~$(G,k)$ is a yes-instance.
Thus, we may assume that~$|X|-|\mathcal{P}|-4<\ell$.
This implies that~$|X|\le\ell+|\mathcal{P}|+3\le 2\ell+4$.

Next, we show that we can also bound the size of~$Y$ in a function only depending on~$\ell$.
To this end, we partition~$Y$ into the set~$Y_U$ of vertices which are incident with only unsafe edges and the set~$Y_S$ of vertices which are incident with at least one safe edge.

First, we show that if~$|Y_S|\ge 2\cdot \binom{|X|}{2}+\ell$, then~$(G,k)$ is a yes-instance. 
Note that since each vertex in~$Y_S$ is incident with at least one safe edge, each vertex of~$Y_S$ can be safely connected with~$X$ by one edge.
Also note that if there is a safe spanning connected subgraph such that for vertices~$x_1,x_2\in X$ and~$y_1,y_2,y_3\in Y_S$ the edges~$\{x_i,y_j\}$ for~$i\in[2]$ and~$j\in[3]$ are contained in this safe spanning connected subgraph, then we can construct a safe spanning connected subgraph with one less edge:
we remove the edges~$\{x_1,y_3\}$ and~$\{x_2,y_3\}$ from this safe spanning connected subgraph and add some safe edge which is incident with~$y_3$.
In other words, we may assume that for each pair~$x_1$ and~$x_2$ of vertices in~$X$ there exists at most $2$~vertices in~$Y_S$ having~$x_1$ and~$x_2$ as neighbors in the safe spanning connected subgraph.
Thus, if~$|Y_S|\ge 2\cdot\binom{|X|}{2}+\ell$, then there exists a safe spanning connected subgraph with at most $2n-4-\ell$~edges.

It remains to show that we can also bound the size of~$Y_U$.
To this end, we construct a new bipartite graph~$H$ as follows:
Let~$Z$ be a vertex set where we add two vertices~$z_{ux}$ and~$z_{xu}$ for each pair of different vertices~$u,x\in X$ such that there exists at least one vertex~$y\in Y_U$ such that~$u,x\in N(y)$.
Note that we add two vertices per pair for technical reasons, discussed later.
Furthermore, we add edges~$\{y,z_{ux}\}$ and~$\{y,z_{xu}\}$ whenever~$u,x\in N(y)$.
Next, apply the following procedure to shrink the size of~$H$:

Initially, set~$i\coloneqq 0$,~$Z_0\coloneqq Z$,~$Y_0\coloneqq Y_U$, and for each~$i$, let~$H_i\coloneqq H[Z_i\cup Y_i]$.
If~$H_i$ has a matching~$M$ saturating~$Z_i$, or~$Z_i$ is empty, we stop the procedure, and denote the endpoints of~$M$ in~$Y_U$ and~$Z$ by~$M_Y$ and~$M_Z$, respectively.
Otherwise, if there is no matching saturating~$Z_i$, then there exists by Hall's theorem~\cite{West00} some~${Z'\subseteq Z_i}$ such that~$|N_{H_i}(Z')|<|Z'|$.
In this case, let~$U_{i+1}\coloneqq Z'$,~$W_{i+1}\coloneqq N_{H_i}(U_i)$, ${Y_{i+1}\coloneqq Y_i\setminus W_{i+1}}$, $Z_{i+1}\coloneqq Z_i\setminus U_{i+1}$, set~$i\coloneqq i+1$, and repeat the above procedure for the new graph~$H_{i}$.

Assume the above procedure ends after $r$~steps.
We have generated disjoint subsets~$U_1, \ldots , U_r$ of~$Z$,  disjoint subsets~$W_1, \ldots ,W_r$ of~$Y_U$, a set~$M_Z=Z\setminus (U_1\cup\ldots\cup U_r)$, and a set~$M_Y\subseteq Y_U\setminus(W_1\cup \ldots\cup W_r)$.
Note that all these sets might be empty.
We define~$Y'\coloneqq Y_U\setminus (M_Y\cup W_1\cup\ldots\cup W_r)$.

Next, we distinguish whether~$M_Y=\emptyset$ or not.
First, we consider the case where~${M_Y=\emptyset}$.
Recall that each vertex in~$Y_U$ is incident with only unsafe edges. 
Thus, we can safely assume that each vertex in~$Y_U$ has at least $2$~neighbors, as otherwise~$(G,k)$ is a trivial no-instance.
Hence from the procedure above, we can conclude that~$Y'=\emptyset$ in this case.
Thus,~$|Y_U|<|Z|\le 2\binom{|X|}{2}$ and hence the size of~$Y$ is bounded by a function only depending on~$\ell$.
In the following, we therefore assume that~$M_Y\ne\emptyset$. 
We use the following claim to shrink the size of~$Y_U$.

\begin{claim} 
\label{claim-bound-size-yu}
There exists a safe spanning connected subgraph~$T$ of~$G$ such that all vertex pairs in~$T-Y'$ are safely connected and for each vertex~$y$ in~$Y'$ the safe spanning connected subgraph~$T$ contains exactly $2$~edges which are incident with~$y$.
\end{claim}
\begin{claimproof}
Let~$T$ be a safe spanning connected subgraph for~$G$. 
Clearly, if each vertex pair in~$T-Y'$ is safely connected, we are done since for each vertex~$y$ in~$Y'$,~$T$ contains exactly $2$~edges which are incident with~$y$.
Otherwise, we do the following to construct a new safe spanning connected subgraph~$T'$ based on~$T$:

First, we show that for two vertices~$u,x\in X$ we either have~$z_{ux},z_{xu}\in M_Z$ or~$z_{ux},z_{xu}\not\in M_Z$.
Suppose exactly one of these vertices, say~$z_{ux}$, is contained in~$M_Z$ and the other vertex~$z_{xu}$ is contained in some~$U_i$.
Note that each vertex in~$Y$ which is adjacent to~$z_{ux}$ is also adjacent with~$z_{xu}$, and vice versa.
Thus,~$z_{ux}$ cannot be matched in~$M_Z$, a contradiction to the definition of~$M_Z$.
Thus, either~$z_{ux},z_{xu}\in M_Z$ or~$z_{ux},z_{xu}\not\in M_Z$.

With the above property at hand, we are now ready to construct a safe spanning connected subgraph~$T'$ based on~$T$:
First, we remove all edges from~$T$ which have an endpoint in~$M_Y$.
Second, let~$M_Z=\{z_1,\ldots, z_{|M_Z|}\}$ such that~$z_i=z_{u_ix_i}$ and~${M_Y=\{y_1,\ldots, y_{M_Z}\}}$ such that~$\{\{z_1,y_1\},\ldots,\{z_{|M_Z|},y_{|M_Z|}\}\}$ form a matching in~$H$.
We add the edges~$\{y_i,u_i\}$ and~$\{y_i,x_i\}$ for each~$i\in[|M_Z|]$ to~$T'$.

Recall that for each two vertices~$u,x\in X$ we either have~$z_{ux},z_{xu}\in M_Z$ or~$z_{ux},z_{xu}\not\in M_Z$.
Observe that by adding the edges~$\{y_i,u_i\}$ and~$\{y_i,x_i\}$ for each~${i\in[|M_Z|]}$, we add all possible connections between vertices in~$X$ which are possible by vertices in~$M_Y$.
Since~$M_Z$ contains $2$~vertices corresponding to each vertex pair matched by~$M_Z$, all connections which are possible via~$M_Y$ are safe.
Thus,~$T'$ is a safe spanning connected subgraph and has the same number of edges as~$T$.
Note that in the above argument, it was necessary to have $2$~vertices in~$Z$ per pair of vertices in~$X$ to ensure that not only one unsafe connection exists.
Since no vertex in~$Z-M_Z$ is adjacent to any vertex in~$Y'$, we conclude that all vertex pairs in~$T'-Y'$ are safely connected and thus the claim follows.
\end{claimproof}

Hence, we can safely remove all vertices in~$Y'$ and obtain that~$|Y_U|\le 2\binom{|X|}{2}$.
Thus, we can obtain an equivalent instance in which the number of vertices is at most \[|X|+|Y_S|+|Y_U|\le |X|+2\binom{|X|}{2}+\ell+2\binom{|X|}{2}=8\ell^2+31\ell+28.\qedhere\]\end{proof}

Now, we can extend \Cref{lem-2-vertex-connected-component} to graphs which are not $2$-vertex connected.
The idea is that if the number of $2$-vertex connected components or $2$-edge connected components is more than~$\ell$, we have a yes-instance.
Then, with the help of \Cref{lem-2-vertex-connected-component} we can show that each $2$-edge connected component is small.

%

\begin{theorem}
\label{thm-fpt-for-2n-ell}
\stue{} can be solved in~${\Oh(\ell^{9\ell}\cdot\poly(n)})$~time, where~$\ell\coloneqq 2n-4-k$.
\end{theorem}

\begin{proof}
We will first reduce any graph to a~$2$-edge connected graph.
Recall that we can assume that each bridge of~$G$ is a safe edge as otherwise~$(G,k)$ is a no-instance.
Note that each such safe bridge~$e$ must be part of a solution and as~$e$ provides a safe connection between its two endpoints, the entire graph~$G$ contains a solution of size~$k$ if and only if the two connected components in~$G- \{e\}$ have solutions of size~$k_1$ and~$k_2$ respectively, where~$k_1+k_2+1 \leq k$.
Moreover, $G$ can contain at most~$\ell + 4$ such bridges or it is a trivial yes-instance as shown next.

We show that if~$k \geq 2n - p$ (and~$n \geq 1$), then~$(G,k)$ is a yes-instance where~$p$ is the number of bridges in~$G$.
We do so via induction over~$p$.
Initially, we consider the case~$p=0$.
If~$n=1$, then there is a solution of size~$0 < 2n$.
If~$n=2$, then the edge between the two vertices is a bridge, a contradiction to the assumption~$p=0$.
If~$n=3$, then all three edges between the three vertices exist and if~$k > 2n = 6$, then taking all triangles in the solution shows that~$(G,k)$ is a yes-instance.
Otherwise, if~$n \geq 4$, then the statement follows from \Cref{ref-ub-solution-ust}.

For the induction step, assume that~$G$ has~$p > 0$ bridges,~$k \geq 2n - p$, and the statement holds for all graphs with at most~$p-1$ bridges.
Let~$e$ be any (safe) bridge in~$G$.
Note that removing~$e$ creates a connected component with~$n_1$ and~$n_2$ vertices, respectively, such that~$n = n_1+ n_2$.
Moreover, the numbers~$p_1$ and~$p_2$ of bridges in the two connected components satisfy~$p = p_1 + p_2 + 1$ and therefore~$\max(p_1,p_2) < p$.
Let~$k_1 = 2n_1 - p_1$ and~$k_2 = 2n_2 - p_2$.
By the induction hypothesis, there exists solutions of size~$k_1$ and~$k_2$ for the two connected components in~$G -\{e\}$, respectively.
Taking these two solutions plus the bridge~$e$ results in a solution of size~$k_1+k_2+1 = 2\cdot(n_1+n_2) - (p_1+p_2) + 1= 2n - (p-1) + 1 = 2n-p$.
This concludes the induction proof.
Now, assume that the number~$p$ of bridges is at least~$\ell + 4$.
By definition of~$\ell$ it holds that~$k = 2n - 4 - \ell \geq 2n - p$ and hence the instance is a yes-instance as claimed.

We next show that each $2$-connected component~$C$ of~$G$ can be replaced by an equivalent instance consisting of at most~$8\ell^3+31\ell^2+28\ell$~vertices.
Note that if this component contains no cut-vertex, then the result follows directly from \Cref{lem-2-vertex-connected-component}.
Hence, let~$v$ be a cut-vertex in~$C$ such that the connected components of~$C-\{v\}$ are~$X_1,\ldots, X_t, T_1, \ldots, T_z$, where each~$X_i$ has size at least~$4$ and each~$T_j$ has size at most~$3$.
We let~$C_i\coloneqq C[X_i\cup\{v\}]$ for~$i\in[t]$,~$C'=C[X_1\cup\ldots\cup X_t\cup\{v\}]$, and~${H_j\coloneqq C[T_j\cup\{v\}]}$ for~${j\in[z]}$.
Let~$T_i$ be a safe spanning connected subgraph for~$C_i$ such that~$T_i$ has the minimal number of edges.
Let~$Q_j$ be similarly defined for~$H_j$.
Note that~$E(T)\coloneqq \bigcup_{i\in[t]}E(T_i)\cup\bigcup_{j\in z}E(Q_j)$ is a safe spanning connected subgraph~$T$ for~$C$ as each vertex is safely connected to~$v$.
Furthermore, let~$\ell_i = 2\cdot|V(C_i)|-4-|E(T_i)|$.
We let~$\ell'\coloneqq \ell_1+\ldots +\ell_t$.
We obtain the following:
\begin{align*}
|E(T)| &= \sum_{i=1}^t(2\cdot|V(C_i)|-4-\ell_i) + \sum_{j=1}^z |E(Q_j)| \\
	&= 2\cdot|V(C')|-4-(2t-3+\ell') + \sum_{j=1}^z (|E(Q_j)|+2\cdot(|V(Q_j)|-1-|V(Q_j)|+1)) \\
	&= 2\cdot|V(C)|-4-(2t-3+\ell') + \sum_{j=1}^z (|E(Q_j)|-2\cdot|V(Q_j)|+2) \\
	&\le 2\cdot|V(C)|-4-(2t-3+\ell') -z
\end{align*}

Note that the last inequality follows from the fact that for connected component of size at most~$3$ we have~$|E(Q_j)|-2\cdot|V(Q_j)|+2<0$.
Note that if~${\ell'+2t-3+z\ge\ell}$, then~$(C,k)$ is a yes-instance of \stue.
Hence, in the following, we may assume that~$\ell'+2t-3+z<\ell$.
This implies that if~${C-\{v\}}$ has at least $\ell$~connected components of size at most~$3$, then~$(C,k)$ is a yes-instance.
By using the above argument iteratively on the $2$-vertex connected components of~$C_i$ (the $2$-vertex connected components of size at least~$4$), we observe that if the number of $2$-vertex connected components in~$C$ is at least~$\ell$, then~$(C,k)$ is a yes-instance.
Thus, in the following, we may safely assume that the number of $2$-vertex connected components is at most~$\ell$.
According to \Cref{lem-2-vertex-connected-component}, we can replace each $2$-vertex connected component of~$C$ by an equivalent instance with at most~$8\ell^2+31\ell+28$~vertices.
Hence,~$C$ can be replaced by an equivalent instance with at most~$8\ell^3+31\ell^2+28\ell$~vertices.
Similarly, if~$G$ contains at least~$\ell+1$ bridges, then~$(G,k)$ is a yes-instance as shown above.
Thus, the entire graph~$G$ can be replaced by a graph with at most
\[(\ell+4)(8\ell^3+31\ell^2+28\ell+1) = 8\ell^4 + 63\ell^3 + 152\ell^2+113\ell+4 \text{ vertices.}\]

Now, we achieve an FPT-time algorithm with respect to~$\ell$ by checking for each subset~$F$ of at most $\ell$~edges whether~$G-F$ is a safe spanning connected subgraph for~$G$.
Note that the number of edges in the constructed graph is in~$\Oh(\ell^8)$.
Hence, we obtain an algorithm with running time~$\Oh({c\ell^8 \choose \ell} \cdot \poly(n)) \subseteq \Oh(\ell^{9\ell}\cdot\poly(n))$.
\end{proof}

\section{Structural Graph Parameters}
\label{sec-graph-params}

In this section, we study two structural graph parameters.
We start by giving an FPT-time algorithm for the parameter treewidth~$\tw$.
Afterwards, we develop a far more intricate FPT-time algorithm for the parameter (vertex-deletion) distance to cluster graphs.
Recall that a cluster graph is a graph in which each connected component is a clique.


\begin{proposition}
\label{prop-fpt-tw}
	\stue{} parameterized by the treewidth~$\tw$ is solvable in~$\Oh(n \cdot 2^{37 \tw \log(\tw)})$ time.
\end{proposition}

\begin{proof}
	Let~$(G=(V,S,U),k)$ be an instance of \stue{} and let~$\tw$ be the treewidth of~$G$.
	First, we compute a nice tree decomposition of~$G$ of width at most~$2\tw+1$ in~$\Oh(2^{\Oh(\tw)} \cdot n)$ time \cite{Kor22}.
	We then compute an optimal solution in a bottom-up fashion.
	To this end, let~$t$ be any node in the nice tree decomposition.
	Recall that for a node~$t$ in the nice tree decomposition, by~$X_t$ we denote the bag associated with~$t$ and by~$Y_t$ we denote the set of all forgotten vertices, that is, the union of all bags associated with descendants of~$t$ except for vertices in~$X_t$.
	The main ingredient for our algorithm are \emph{backbones}.
	Intuitively, a backbone describes all connections between vertices in~$X_t$ that are ensured by edges with at least one endpoint in~$Y_t$.
	One can think of them as a compact representation of a partial solution.
	A partial solution~$H$ is a subgraph of~$G$ with vertex set~$Y_t \cup X_t$ and where each edge in~$H$ has at least one endpoint in~$Y_t$.
	In order to keep the number of possible backbones small enough, we do not enumerate all partial solutions but instead contract all vertices that are pairwise safely connected in the partial solution into a single vertex.
	Moreover, (for technical reasons) we further contract all degree-two vertices in this representation unless they represent vertices in~$X_t$.
	We represent these contracted graphs by \emph{backbones}.
	Formally, a backbone~$B = (\ell, c,\mathcal{T})$ is a triple, where~$\ell \in [|X_t|]$ is an integer,~$c \colon X_t \rightarrow [\ell]$ is a surjective vertex-coloring function, and~$\mathcal{T}=\{T_1,T_2,\ldots,T_p\}$ is a set of trees such that
		\begin{enumerate}
			\item the set of leaves of~$T_i$ is a subset of~$[\ell]$ for each~$i\in [p]$,
			\item the tree~$T_i$ contains no degree-2 vertices for each~$i \in [p]$,
			\item two trees~$T_i$ and~$T_j$ share no vertices outside of~$[\ell]$, and
			\item the union of all~$T_i$ is acyclic.
		\end{enumerate}
	Given a backbone~$B$, we denote the graph containing the union of all trees in~$\mathcal{T}$ (and all vertices in~$[\ell]$) by~$G_B$.
	In the aforementioned intuition where~$G_B$ is the contraction of a partial solution, the set~$[\ell]$ can be seen as the set of vertices in~$G_B$ that represent at least one vertex in~$X_t$.
	A vertex~$v$ in~$X_t$ is represented by vertex~$j \in [\ell]$ if and only if~$c(v) = j$.
	The vertices in the trees in~$\mathcal{T}$ represent the vertices in~$Y_t$ and the edges represent the unsafe connections between the different safely connected components.
	Note that any cycle in this representation would imply that all represented vertices belong to the same safely connected component.
	Since we did not further contract the partial solution, we may assume that the union of all~$T_i \in \mathcal{T}$ is acyclic.
	For the same reason, we may assume that no two trees~$T_i$ and~$T_j$ share a vertex outside of~$[\ell]$, as otherwise we would have merged them into a single tree.
	Finally, if a tree~$T_i \in \mathcal{T}$ contains a degree-2 vertex, say~$v$, with neighbors~$u$ and~$w$, then~$v$ is safely connected to both~$u$ and~$w$ once the latter two are safely connected to one another.
	Thus, we can indeed contract degree-two vertices without changing the outcome.
	Furthermore, note that each edge in~$G_B$ represents an unsafe edge between two safely connected components.
	We therefore define all edges in~$G_B$ to be unsafe edges.
	
	We say that a partial solution~$H$ \emph{implements}~$B$ if~$G_B$ is the result of contracting all safely connected components in the partial solution into a single vertex and then contracting all degree-two vertices in this representation unless they represent vertices in~$X_t$.
	More formally, $H$ implements~$B$ if the following conditions are met:
	\begin{enumerate}
		\item $c(u) = c(v)$ for two vertices~$u,v \in X_t$ if and only if~$u$ is safely connected to~$v$ in~$H$,
		\item contracting all safely connected components in~$H$ into single vertices and then contracting degree-two vertices (replacing degree-two vertices by edges unless the vertex represents at least one vertex in~$X_t$) results in the union of the set of trees in~$\mathcal{T}$, and
		\item each vertex in~$Y_t$ is safely connected to~$X_t$.
	\end{enumerate}
	Therein, we say that a vertex is safely connected to a set~$X_t$ of vertices if it safely connected to a vertex in~$X_t$ or unsafely connected to at least two vertices in~$X_t$.
	The last requirement is again necessary to keep the number of possible backbones small enough.
	Note however, that since~$X_t$ is a separator separating~$Y_t$ from the rest of the graph, if a vertex~$v \in Y_t$ is not safely connected to~$X_t$, then a partial solution implementing~$B$ can never lead to a solution as~$v$ cannot be safely connected to the rest of the graph.
	We therefore do not consider such backbones.
	
	We now solve \stue{} using dynamic programming as follows.
	For each node~$t$ in the nice tree decomposition and each possible backbone~$B = (\ell,c,\mathcal{T})$, we store in a table~$\DP$ the minimum size of a partial solution that implements~$B$.
	In order to do so, we first define two operations.
	The first one is the \emph{union}~$B_1 \oplus B_2$ of two backbones~$B_1$ and~$B_2$.
	Informally speaking, it is the backbone implemented by the union of two partial solutions~$H_1$ and~$H_2$ such that~$H_1$ implements~$B_1$ and~$H_2$ implements~$B_2$.
	For the formal definition, let~$B_1 = (\ell_1,c_1,\mathcal{T}_1)$ and let~$B_2 = (\ell_2,c_2,\mathcal{T}_2)$.
	To avoid confusion, we denote the vertices in~$[\ell_1]$ in~$G_{B_1}$ by~$v_1,v_2,\ldots,v_{\ell_1}$ and the vertices in~$[\ell_2]$ in~$G_{B_2}$ by~$u_1,u_2,\ldots,u_{\ell_2}$.
	We define a graph~$G'$ as follows.
	We start with the disjoint union of~$G_{B_1}$ and~$G_{B_2}$.
	We then add a safe edge between~$v_i$ and~$u_j$ if there is a vertex~$w \in X_t$ such that~$c_1(w) = i$ and~$c_2(w) = j$.
	To finish the construction of~$G'$, contract all safely connected components into single vertices and contract any degree-two vertex unless it is a vertex~$v_i$ or~$u_i$ for some~$i$.
	Let~$B = (\ell,c,\mathcal{T}) = B_1 \oplus B_2$ be the backbone such that~$G_B = G'$.
	Formally, let~$c$ map each vertex~$w \in X_t$ to the vertex~$c_1(w)$ or the vertex that~$c_2(w)$ was contracted into.
	Let~$\ell$ be the size of the image of~$c$, and let~$\mathcal{T}$ be the set of trees in~$G'$ with the image of~$c$ as leaves.
	Note that~$G'$ (and therefore all graphs in~$\mathcal{T}$) are indeed acyclic as any cycle would imply a safely connected component and would hence be contracted into a single vertex.
	
	We call the second operation~$\movev$.
	Intuitively, given a backbone~$B = (\ell,c,\mathcal{T})$ for a node~$t$ and a vertex~$v \in X_t$, it returns an equivalent backbone for~$X_t \setminus \{v\}$.
	We distinguish between two cases.
	If~$c(u) = c(v)$ for any vertex~$u \in X_t \setminus \{v\}$, then the operation does nothing.
	Otherwise, it checks whether~$v$ is safely connected to~$X_t \setminus \{v\}$.
	If not, then the operation fails.
	Otherwise, it simply renames the vertex~$c(v)$ to something not contained in~$[\ell]$.
	In terms of the backbone~$B$, it removes~$v$ from the domain of~$c$ (for notational ease we will assume that~$c(v) = \ell$ previously to avoid relabeling), reduces~$\ell$ by one, and modifies~$\mathcal{T}$ as follows.
	Let~$\mathcal{T}' \subseteq \mathcal{T}$ be the set of all trees in~$\mathcal{T}$ that contain~$c(v)$.
	We remove each tree in~$\mathcal{T'}$ from~$\mathcal{T}$ and add a new tree to it.
	The new tree is the result of merging all trees in~$\mathcal{T'}$ as follows.
	We start by introducing a new vertex~$v'$ and in each tree~$T \in \mathcal{T'}$, we replace~$c(v)$ by~$v'$.
	Note that since we assume the graph consisting of the union of all graphs in~$\mathcal{T}$ to be acyclic, the newly created graph is indeed a tree.
	If the newly introduced vertex~$v'$ has degree two, then we contract it.

	\paragraph*{Dynamic program}
	Note that for the root node~$r$ in the nice tree decomposition, there is only a single valid backbone~$B_r=(1,c,\emptyset)$ where~$c(v)=1$ for the unique vertex~$v$.
	By definition of the dynamic program,~$\DP[r,1,c,\emptyset]$ stores the minimum cost of a subgraph~$H$ of~$G$ with vertex set~$V$ such that each vertex in~$V \setminus \{v\}$ is safely connected to~$v$.
	Thus, if we can correctly compute the whole table~$\DP$, then we can solve \stue{} by checking whether~$\DP[r,1,c,\emptyset]\leq k$.
	
	We proceed with the description of how to fill the table~$\DP$.
	To this end, let~$(t,\ell,c,\mathcal{T})$ be an arbitrary entry and suppose that all table entries for descendants of~$t$ have already been computed.
	We make a case distinction over the type of~$t$.
	
	If~$t$ is an introduce node, then let~$v$ be the vertex introduced by~$t$.
	If~$t$ is a leaf in the nice tree decomposition (that is, it has no children), then there is only one valid backbone~$B=(1,c,\emptyset)$ with~$c(v)=1$ for~$t$.
	Hence~$\DP[t,1,c,\emptyset]=0$.
	If~$t$ has a child~$t'$ in the nice tree decomposition, then we note that there is no edge in~${S \cup U}$ between~$v$ and a vertex in~$Y_t$.
	Hence, if~$\ell=1$,~$c(v) = c(u)$ for some~$u \neq v$, or~$c(v)$ is contained in a tree in~$\mathcal{T}$, then the backbone cannot be implemented and we can set~$\DP[t,\ell,c,\mathcal{T}] = \infty$.
	Otherwise, we assume for the sake of notational ease that~${c(v) = \ell}$.
	Then, $\DP[t,\ell,c,\mathcal{T}] = \DP[t',\ell-1,c',\mathcal{T}]$, where~$c'(u) = c(u)$ for all~$u \in X_{t'}$.
	
	If~$t$ is a join node, then let~$t_1$ and~$t_2$ be the children of~$t$ in the nice tree decomposition.
	We set
	
	$$\DP[t,\ell,c,\mathcal{T}] = \min\limits_{\ell_1,c_1,\mathcal{T}_1,\ell_2,c_2,\mathcal{T}_2} \DP[t_1,\ell_1,c_1,\mathcal{T}_1] + \DP[t_2,\ell_2,c_2,\mathcal{T}_2] + \mathbf{1}_\oplus,$$
	where~$\mathbf{1}_\oplus = 0$ if~$(\ell,c,\mathcal{T}) = (\ell_1,c_1,\mathcal{T}_1) \oplus (\ell_2,c_2,\mathcal{T}_2)$ and~$\mathbf{1}_\oplus = \infty$ otherwise.
	
	If~$t$ is a forget node, then let~$v$ be the vertex forgotten by~$t$ and let~$t'$ be the child of~$t$ in the nice tree decomposition.
	Let~${S_v = \{\{u,v\} \in S \mid u \in X_t\}}$ and~${U_v = \{\{u,v\} \in U \mid u \in X_t\}}$ be the set of safe and unsafe edges between~$v$ and vertices in~$X_t$, respectively.
	For subsets~$A \subseteq S_v$, let~$c_A$ be defined as~$c_A(u)=1$ if~$\{u,v\} \in A$ and all other vertices in~$X_t \setminus \{v\}$ are assigned unique consecutive numbers.
	Moreover, let~$\ell_A$ be the size of the image of~$c_A$.
	For subsets~$A \subseteq S_v$ and~$B \subseteq U_v$, let~$\mathcal{T}_B$ be the set of trees where for each edge~$\{u,v\} \in B$, the set~$\mathcal{T}_B$ contains a tree consisting of vertices~$c_A(u)$ and~$c_A(v)$ and an edge between the two vertices.
	We set

	$$\DP[t,\ell,c,\mathcal{T}] = \min\limits_{\ell',c',\mathcal{T'},A,B} \DP[t',\ell',c',\mathcal{T'}] + |A| + |B| + \mathbf{1}^v_\oplus,$$ 
	where~$\mathbf{1}^v_\oplus = 0$ if~$(\ell,c,\mathcal{T}) = \movev((\ell',c',\mathcal{T'}) \oplus (\ell_A,c_A,\mathcal{T}_B))$ and~$\mathbf{1}^v_\oplus = \infty$ otherwise.
	
	\paragraph*{Correctness of $\oplus$ and~$\movev$}
	It remains to prove the correctness of our algorithm and to analyze the running time.
	To this end, we first prove that the operations~$\oplus$ and~$\movev$ work as intended.
	We start with~$\oplus$ and we show that for any node~$t$ in the nice tree decomposition, if there are edge-disjoint graphs~$H_1$ and~$H_2$ that implement two backbones~$B_1 = (\ell_1,c_1,\mathcal{T}_1)$ and~$B_2 = (\ell_2,c_2,\mathcal{T}_2)$ for~$t$, then~$H = H_1 \cup H_2$ implements~$B = (\ell,c,\mathcal{T}) = B_1 \oplus B_2$.\footnote{We remark that~$B_1$ and~$B_2$ are technically not backbones (at least not for node~$t$) as we will not guarantee that both~$H_1$ and~$H_2$ safely connect all forgotten vertices to~$X_t$ but only that each forgotten vertex is safely connected to~$X_t$ by either~$H_1$ or~$H_2$.} 
	First, we show that~$c(u) = c(v)$ for two vertices~$u,v \in X_t$ if and only if~$u$ and~$v$ are safely connected by~$H$.
	To this end, assume first that~$u$ and~$v$ are safely connected by~$H$.
	Then, there are two $u$-$v$-paths~$P_1$ and~$P_2$ such that~$P_1$ and~$P_2$ do not overlap in any unsafe edges.
	By construction, the vertices~$v_{c_1(u)}$ and~$v_{c_1(v)}$ belong to the same safely connected component in~$G_B$.
	Thus,~$c(u) = c(v)$ holds.
	Now assume that~$c(u) = c(v)$ for two vertices~$u,v \in X_t$.
	By construction,~$v_{c_1(u)}$ and~$v_{c_1(v)}$ are contracted into the same vertex in~$G_B$.
	Note that we only contract vertices in~$G_B$ if they are in the same safely connected component.
	All unsafe edges in~$G_B$ correspond by definition to unsafe edges in~$H_1$ or~$H_2$.
	We only add a safe edge to~$G_B$ if a vertex in~$G_{B_1}$ represents a vertex in~$G$ that is also represented by a vertex in~$G_{B_2}$.
	Note that in this case the two respective safely connected components indeed merge into a single safely connected component in~$G_B$.
	Thus, $u$ and~$v$ are safely connected by~$H = H_1 \cup H_2$.
	Second, we show that contracting every safely connected component in~$H$ into a single vertex and removing degree-two vertices results in the union of all trees in~$\mathcal{T}$.
	This follows from the fact that~$H_1$ implements~$B_1$,~$H_2$ implements~$B_2$, $H_1$ and~$H_2$ are edge-disjoint and the union of all trees in~$\mathcal{T}$ is the result of contracting all safely connected components and degree-two vertices in the disjoint union of~$G_{B_1}$ and~$G_{B_2}$.
	Finally, it remains to show that each vertex in~$Y_t$ is safely connected to a vertex in~$X_t$ or unsafely connected to two vertices in~$X_t$.
	This holds trivially, since each vertex in~$Y_t$ is connected to~$X_t$ by~$H_1 \subseteq H$\footnote{Or in the case of a join node, each vertex in~$Y_t$ is contained in~$Y_{t_1}$ or in~$Y_{t_2}$, where~$t_1$ and~$t_2$ are the children of~$t$ and~$B_i$ is a backbone for~$t_i$ for each~$i\in[2]$. Then, each vertex in~$Y_{t_i}$ is safely connected to~$X_t$ in~$H_i$ and therefore also in~$H$.}.
	
	Next, we show that~$\movev$ works as intended, that is, a partial solution~$H$ implements a backbone~$B$ for a node~$t$ with bag~$X_t$ such that~$v \in X_t$ is safely connected to~$X_{t} \setminus \{v\}$ if and only if~$H$ implements the backbone~$\movev(B)$ for the forget node~$t'$ with~$X_{t'} = X_t \setminus \{v\}$ that is a parent of~$t$.
	First, assume that~$H$ implements~$B$ and~$v$ is safely connected to~$X_t$.
	If~$v$ is safely connected to a vertex~$u \in X_t$, then~$c(v) = c(u)$ and~$\movev$ does nothing.
	Hence, $H$ implements~$\movev(B) = B$ by assumption.
	If~$v$ is not safely connected to a vertex in~$X_t$, then, since it is safely connected to~$X_t$, it is unsafely connected to at least two vertices in~$X_t$.
	Hence, when moving~$v$ out of~$X_t$, we create one new safely connected component that does not contain any vertices in~$X_t \setminus \{v\}$.
	All trees that previously had a leaf in this connected component are now merged together as the new connected component does not contain a vertex in~$X_t \setminus \{v\}$.
	Note that this does not change the graph~$G_B$ other than renaming the vertex~$v$ and possibly contracting it if it has degree two.
	As argued above, contracting degree-two vertices does not matter and therefore~$H$ implements~$\movev(B)$.
	Conversely, assume that~$H$ implements~$\movev(B)$.
	By definition, each vertex in~$Y_{t'}$ is safely connected to~$X_{t'}$.
	Since~$v \in Y_{t'}$,~$v$ is safely connected to~$X_{t'} = X_{t} \setminus \{v\}$ and~$\movev(B)$ therefore does not fail.
	It remains to show that~$H$ implements~$B$.
	Again, if~$v$ is safely connected to a vertex~$u \in X_{t'}$, then~$\movev(B) = B$ and the claim holds by assumption.
	Otherwise, $G_B$ and~$G_{\movev(B)}$ only potentially differ in the contraction of a degree-two vertex.
	
	\paragraph*{Correctness of $\DP$}
	With these two ingredients, we can show the correctness of our computation of~$\DP$, that is, that~$\DP$ stores the size of a minimum set~$H$ of edges that implements a given backbone.
	The basic idea is to observe that we only use~$\oplus$ on backbones whose implementations are necessarily edge-disjoint and the fact that we ensure that each forgotten vertex is safely connected to the rest of the graph once all vertices in~$X_t$ are safely connected to one another (which is guaranteed in the root node).
	We make an induction proof over the height of the node~$t$ in the nice tree decomposition and proceed with a case distinction over the type of~$t$ to show that~$\DP[t,\ell,c,\mathcal{T}]$ is correctly computed.
	Let~$B = (\ell,c,\mathcal{T})$.
	If~$t$ is an introduce node, then let~$v$ be the vertex introduced by~$t$.
	If~$t$ is a leaf, then there is only one valid backbone whose table entry is correctly computed as argued above.
	If~$t$ is not a leaf, then let~$t'$ be its child in the nice tree decomposition.
	Since there is by definition of~$X_{t'}$ no edge between~$Y_{t'} = Y_t$ and~$v$, we only need to consider backbones in which~$c(v) \neq c(u)$ for each~$u \in X_{t} \setminus \{v\}$ and~$c(v)$ is not a leaf in any tree in~$\mathcal{T}$.
	Thus, each such backbone has a one-to-one correspondence with a backbone for~$t'$ and computing~$\DP[t,\ell,c,\mathcal{T}] = \DP[t',\ell-1,c',\mathcal{T}]$, where~$c'(u) = c(u) < \ell$ for all~$u \in X_{t'}$ and~$c(v) = \ell$ is correct.

	If~$t$ is a join node, then let~$t_1$ and~$t_2$ be its two children.
	Let~$H$ be a partial solution of minimal size that implements~$B$.
	We show that~${\DP[t,\ell,c,\mathcal{T}] = |H|}$.
	Note that since~$Y_{t_1}$ and~$Y_{t_2}$ are disjoint, there is no edge between a vertex in~$Y_1$ and a vertex in~$Y_2$, and~${Y_{t} = Y_{t_1} \cup Y_{t_2}}$.
	Hence,~$H$ can be partitioned into two set~$H_1$ and~$H_2$ such that each edge in~$H_1$ has an endpoint in~$Y_1$ and each edge in~$H_2$ has an endpoint in~$Y_2$.
	Let~${B_1 = (\ell_1,c_1,\mathcal{T}_1)}$ and~${B_2 = (\ell_2,c_2,\mathcal{T}_2)}$ be the two backbones such that~$H_1$ implements~$B_1$ in~$t_1$ and~$H_2$ implements~$B_2$ in~$t_2$.
	Note that~$B_1$ and~$B_2$ are indeed backbones as each forgotten vertex is safely connected to~$X_t$ by~$H$ and hence is also safely connected to~$X_t$ by~$H_1$ or~$H_2$.
	Hence,~${\DP[t,\ell,c,\mathcal{T}] \leq \DP[t_1,\ell_1,c_1,\mathcal{T}_1] + \DP[t_2,\ell_2,c_2,\mathcal{T}_2]}$.
	Moreover,~$B = B_1 \oplus B_2$ and therefore~$\mathbf{1}_\oplus = 0$.
	Since by the induction hypothesis,~$\DP[t_1,\ell_1,c_1,\mathcal{T}_1] = |H_1|$ and~$\DP[t_2,\ell_2,c_2,\mathcal{T}_2] = |H_2|$, we have
	$$\DP[t,\ell,c,\mathcal{T}] = \min\limits_{\ell_1,c_1,\mathcal{T}_1,\ell_2,c_2,\mathcal{T}_2} \DP[t_1,\ell_1,c_1,\mathcal{T}_1] + \DP[t_2,\ell_2,c_2,\mathcal{T}_2] + \mathbf{1}_\oplus \leq |H_1| + |H_2| = |H|.$$
	To see that also $\DP[t,\ell,c,\mathcal{T}] \geq |H|$, note that~$\mathbf{1}_\oplus$ guarantees that~$\DP$ is only updated if the algorithm finds a valid partial solution and~$H$ is by definition the smallest such partial solution.
	
	If~$t$ is a forget node, then let~$t'$ be its child and let~$v$ be the vertex that is forgotten by~$t$.
	Let again~$H$ be a partial solution of minimal size that implements~$B$.
	We show that~$\DP[t,\ell,c,\mathcal{T}] = |H|$.
	Let~$H' \subseteq H$ be the set of all edges in~$H$ that do not have~$v$ as an endpoint.
	Let~$A,B \subseteq H \setminus H'$ be the set of safe and unsafe edges incident to~$v$ in~$H$, respectively.
	Let~$B' = (\ell',c',\mathcal{T}')$ be the backbone such that~$H'$ implements~$B'$ in~$t'$.
	Note that~$\DP[t,\ell,c,\mathcal{T}] \leq \DP[t',\ell',c',\mathcal{T}'] + |A| + |B|$.
	Moreover, we have by construction that~${(\ell,c,\mathcal{T}) = \movev((\ell',c',\mathcal{T'}) \oplus (\ell_A,c_A,\mathcal{T}_B))}$ and therefore~$\mathbf{1}^v_\oplus = 0$.
	Hence by induction hypothesis and the recursive formula for~$\DP$, we have~$\DP[t,\ell,c,\mathcal{T}] \leq |H'| + |A| + |B| = |H|$.
	To see that also $\DP[t,\ell,c,\mathcal{T}] \geq |H|$, note that~$\mathbf{1}^v_\oplus$ guarantees that~$\DP$ is only updated if the algorithm finds a valid partial solution and~$H$ is by definition the smallest such partial solution.
	
	\paragraph*{Running time}
	Finally, we analyze the running time.
	To this end, let~$z = 2 \tw + 2$.
	First, we compute a tree decomposition of width at most~$z - 1$ in~$2^{\Oh(\tw)} \cdot n$ time \cite{Kor22}.
	Next, we transform the tree decomposition into a nice tree decomposition of width at most~$z-1$ in~$\Oh(n \cdot \tw^2)$ time~\cite{Kloks94}.
	We continue by showing that there are at most~$(2z)^{3z+1}$ possible backbones for a node~$t$.
	There are~$z$ possibilities for the value of~$\ell$ and the number of possible coloring functions is at most~$z^z$.
	The number of possible sets of trees is at most~$(2z)^{2z}$ as shown next.
	Since the union of all trees is acyclic and contains no vertices of degree at most two outside of~$[\ell]$, the union of all trees is a forest with at most~$2z$ vertices.
	By Cayley's formula, the number of such forests is at most~$(2z)^{2z}$ \cite{Tak90}.

	Before we can finally analyze the time it takes to fill the table~$\DP$, we note that it takes~$\Oh(z)$ time to perform the operations~$\oplus$ and~$\movev$ as all the graphs involved contain~$\Oh(z)$ vertices and edges.
	The number of entries in~$\DP$ is in~$\Oh(n \cdot (2z)^{3z+1})$ as there are~$\Oh(n)$ nodes in the nice tree decomposition and for each of them there are at most~$(2z)^{3z+1}$ possible backbones.
	The time to compute one table entry depends on the type of~$t$.
	Since the time to compute a table entry for a join node dominates the running time for all other cases, we focus on this case here.
	We iterate over all possible combinations of backbones for the two children and for each combination, we compute the~$\oplus$ operation.
	Hence, the time needed to compute an entry is in~$\Oh((2z)^{6z+3})$.
	Thus, we can fill the whole table~$\DP$ in overall~$\Oh(n \cdot (2z)^{9z+4})$ time.
	Substituting~${z=2\tw+2}$ and observing that the time to fill the table~$\DP$ dominates the running time of all other computations yields an overall running time of~${\Oh(n \cdot (4\tw+4)^{18\tw + 22}) \subseteq \Oh(n \cdot 2^{37 \tw \log (\tw)})}$.
	This concludes the proof.
\end{proof}

We continue with the algorithm for distance to cluster graphs.

\begin{theorem}
\label{thm-fpt-dist-cluster}
	\stue{} parameterized by the (vertex-deletion) distance~$\ell$ to cluster graphs can be solved in~$f(\ell) \cdot n^3$ time for some computable function~$f$.
\end{theorem}

\begin{proof}
We start by computing a set~$K$ of~$\ell$ vertices in~$\Oh(3^\ell \cdot n^3)$ time such that~${G' = G[V \setminus K]}$ is a cluster graph as follows. 
Note that a graph~$G$ is a cluster graph if and only if it does not contain an induced~$P_3$, that is, three vertices~$a,b,$ and~$c$ with~$\{a,b\},\{b,c\} \in E(G)$ and~$\{a,c\} \notin E(G)$.
Hence, we can find an induced~$P_3$ in~$\Oh(n^3)$~time if it exists and then branch on which of the three vertices in the~$P_3$ to include in~$K$.
Note that~$K$ needs to contain at least one of the three vertices and the resulting search tree has therefore depth~$\ell$ and size~$3^\ell$.

We next give a few definitions required to give a more detailed description of the algorithm afterwards.
We call a subgraph of a solution~$H=(V_H,S_H,U_H)$ a \emph{backbone} if it contains all vertices in~$K$ and each pair of vertices in~$V_H$ is safely connected in~$H$.
See \cref{fig:backbone} for an example.
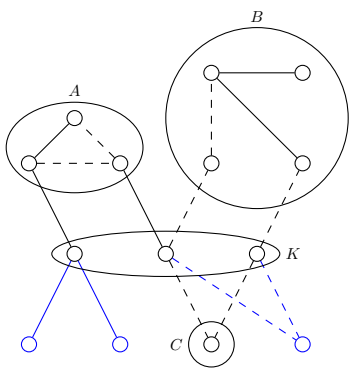
\begin{figure}[t]
	\centering
	\begin{tikzpicture}[scale=0.6,every node/.style={scale=0.6}]
		\node[circle,draw] at (0,0) (a) {};
		\node[circle,draw] at (2,0) (b) {};
		\node[circle,draw] at (4,0) (c) {};
		
		\node[circle,draw] at (-1,2) (d) {};
		\node[circle,draw] at (1,2) (e) {};
		\node[circle,draw] at (0,3) (f) {};
		
		\node[circle,draw] at (3,2) (g) {};
		\node[circle,draw] at (3,4) (h) {};
		\node[circle,draw] at (5,2) (i) {};
		\node[circle,draw] at (5,4) (j) {};
		
		\node[circle,draw=blue] at (-1,-2) (k) {};
		
		\node[circle,draw=blue] at (1,-2) (l) {};
		
		\node[circle,draw] at (3,-2) (m) {};
		
		\node[circle,draw=blue] at (5,-2) (n) {};
		
		\draw[blue] (a) to (k);
		\draw[blue] (a) to (l);
		\draw (a) to (d);
		\draw (d) to (f);
		\draw[dashed] (d) to (e);
		\draw[dashed] (e) to (f);
		\draw (b) to (e);
		\draw[dashed] (b) to (m);
		\draw[dashed] (c) to (m);
		\draw[dashed,blue] (b) to (n);
		\draw[dashed,blue] (c) to (n);
		\draw[dashed] (b) to (g);
		\draw[dashed] (c) to (i);
		\draw[dashed] (g) to (h);
		\draw (i) to (h);
		\draw (h) to (j);
		
		\node[ellipse, draw, minimum width=5cm, minimum height=1cm, label=right:$K$] at(2,0) {};
		\node[ellipse, draw, minimum width=3cm, minimum height=2cm, label=above:$A$] at(0,2.35) {};
		\node[ellipse, draw, minimum width=4cm, minimum height=4cm, label=above:$B$] at(4,3) {};
		\node[ellipse, draw, minimum width=1cm, minimum height=1cm, label=left:$C$] at(3,-2) {};
	\end{tikzpicture}
	\caption{An example of a solution. Safe edges are depicted with solid lines and unsafe edges are depicted with dashed lines. The modulator~$K$ and three connecting components~$A$,~$B$, and~$C$ of a minimal backbone are drawn in black. All other vertices (drawn in blue) are not part of the minimal backbone as all vertices in~$K$ are already safely connected in the black subgraph. The connecting component~$A$ is a cyclic component and~$B$ and~$C$ are usual connecting components.}
	\label{fig:backbone}
\end{figure}%
Given such a backbone~$H=(V_H,S_H,U_H)$, a \emph{connecting component} is a connected component in~$H'=H[V_H\setminus K]$.
We say that a backbone is \emph{minimal} if the removal of any connecting component results in some pair of vertices in~$K$ being not safely connected anymore.
We distinguish between two types of connecting components: \emph{cyclic} and \emph{usual}. 
A cyclic connecting component is a cycle (with some connections to vertices in~$K$).
Note that the number of edges in a cycle equals its number of vertices.
A usual connecting component is a tree.
Its number of edges is one less than its number of vertices but if it contains unsafe edges, then not all vertices are safely connected within the connecting component.
Observe that we can indeed assume that each connecting component is either usual or cyclic as any connecting component that is not a tree contains at least as many edges as vertices.
Hence, we can replace such a connecting component by any cycle through all of its vertices.
Such a cycle exists within any clique of size at least three (any permutation of the vertices results in a cycle and any connecting component inside a clique of size at most two is a tree) and any pair of vertices in a cyclic connecting component is safely connected.


Next, we distinguish between three types of cliques in~$G'$.
To this end, let~$C$ be a connected component in~$G'$, that is, a clique in the cluster graph.
The three types are based on the edges between~$C$ and~$K$ (note that all edges between vertices in~$C$ and the rest of the graph are to vertices in~$K$) and on the connected components within~$C$ if we ignore all unsafe edges.
We call such components \emph{strong components}.
A clique~$C$ is \emph{strong}, if each strong component in~$C$ contains at least one vertex with a safe edge to a vertex in~$K$.
In this case, we can add~$C$ to the backbone using~$|C|$ (safe) edges (a maximal spanning forest plus for each strong component one edge connecting it to~$K$).
The second type, we call \emph{weak cliques}.
A weak clique is not a strong clique but it is connected to~$K$ by a safe edge or by two unsafe edges with different endpoints in~$C$.
Since weak cliques are not strong cliques, there is some strong component~$L$ in them, which is not connected to~$K$ via only safe edges.
Hence, to safely connect this strong component~$L$ to~$K$, we require at least~$|L|+1$ edges.
Since we require at least~$|C \setminus L|$ edges to connect the remaining vertices, we require at least~$|C|+1$ edges to safely connect~$C$ to the backbone.
For weak cliques,~$|C|+1$ edges are sufficient as we can use a) a safe edge between~$C$ and~$K$ and a Hamiltonian cycle in~$C$ or b) two unsafe edges to different vertices in~$C$ and any Hamiltonian path between the two endpoints within~$C$.
We call the third type of clique \emph{singletons}.
A singleton is only connected to~$K$ by unsafe edges and all of these edges have the same endpoint in~$C$.
Note that if there is at most one unsafe edge between~$K$ and~$C$ (and assuming that~$K \neq \emptyset$), then there cannot be a solution as~$C$ cannot be safely connected to~$K$.
We call them singletons because we can reduce such a clique to a single vertex in a preprocessing step.
Since all connections to~$K$ are through one vertex~$v \in C$, we have to safely connect all vertices in~$C$ to~$v$.
This can either be done via a spanning tree consisting only of safe edges (if such a tree exists) or via any Hamiltonian cycle otherwise\footnote{The Hamiltonian cycle exists if there are at least three vertices in~$C$. If~$|C|=2$ and there is only an unsafe edge between the two vertices, then there cannot be a solution and we can return \emph{no}.}.
Which case applies can be checked in linear time by checking whether there is exactly one strong component in~$C$.


We are now in a position to describe the algorithm.
First, we guess which edges between vertices in~$K$ belong to a solution~\SSS{} and the number~$p$ of connecting components in a minimal backbone~$H$ of~\SSS.
Note that~$p \leq 2\ell$ since any connecting component in~$H$ provides an (unsafe) connection between two vertices in~$K$ and any cycle of unsafe connections implies also safe connections.
Hence in the worst case, each cycle is of length two and we require~$2\ell-2$ connections to implement a ``safe spanning tree'' between the vertices in~$K$. 
Next, we guess the structure of~$H$, that is, for each connecting component~$P$ in~$H$, we guess the following (see also~\Cref{fig:guessInterface}). 
\begin{enumerate}
	\item Which vertices in~$K$ are adjacent to vertices in~$P$ in~\SSS?
	\item How large is the set~$Y$ of vertices in~$P$ that are neighbors to vertices in~$K$ in~\SSS?\footnote{Note that~$|Y| \leq 2\ell$ by the same argument that shows~$p \leq 2\ell$.}
	\item Which edges between~$K$ and~$Y$ are contained in \SSS?\footnote{Note that we guess something like ``there are three vertices~$y_1,y_2,y_3 \in Y$ and the solution contains the safe edge~$\{u,y_1\}$ and the unsafe edges~$\{v,y_2\}, \{v,y_3\}$, and~$\{w,y_3\}$''. However, we do not guess which vertex in the input graph is a vertex in~$Y$. For an example, we refer to \Cref{fig:guessInterface}.}
	\item Which pairs of vertices in~$Y$ are safely connected within~$Y$? (That is, a partition of the vertices in~$Y$)
	\item Is~$P$ a usual or a cyclic component?
\end{enumerate}
Moreover, we guess which connecting components of the minimal backbone are contained in the same clique in~$G'$, that is, we guess a partition of the~$p$ connecting components.
\begin{figure}[t]
	\centering
	\begin{tikzpicture}[scale=0.6,every node/.style={scale=0.6}]
		\def\d{10}
		\node[circle,draw,label=$u$] at(1,4) (u) {};
		\node[circle,draw,label=$v$] at(3,4) (v) {};
		\node[circle,draw,label=$w$] at(5,4) (w) {};
		
		\node[circle,draw,label=right:$a$] at(0,2) (a) {};
		\node[circle,draw,label=right:$b$] at(2,2) (b) {};
		\node[circle,draw,label=left:$c$] at(4,2) (c) {};
		\node[circle,draw,label=left:$d$] at(6,2) (d) {};
		
		\node[circle,draw] at(1,0) (e) {};
		\node[circle,draw] at(3,0) (f) {};
		\node[circle,draw] at(5,0) (g) {};
		
		\draw (u) to (a);
		\draw (u) to (b);
		\draw[dashed] (v) to (b);
		\draw[dashed] (v) to (c);
		\draw[dashed] (v) to (d);
		\draw[dashed] (w) to (b);
		\draw[dashed] (w) to (c);
		\draw[dashed] (w) to (d);
		\draw (a) to (e);
		\draw (b) to (e);
		\draw (c) to (f);
		\draw (e) to (f);
		
		\node[ellipse, draw, minimum width=5.5cm, minimum height=1.5cm, label=right:$K$] at(3,4.1) {};
		\node[ellipse, draw, minimum width=8cm, minimum height=4cm, label=right:$C$] at(3,1) {};

		\node[circle,draw,label=below:$u$] at(\d,4) (u2) {};
		\node[circle,draw,label=$v$] at(\d,2) (v2) {};
		\node[circle,draw,label=$w$] at(\d,0) (w2) {};
		
		\node[circle,draw,label=below:$y_1$] at(\d+2,4) (a2) {};
		\node[circle,draw,label=above:$y_2$] at(\d+2,2) (b2) {};
		\node[circle,draw,label=right:$y_3$] at(\d+2,0) (c2) {};
		
		\draw (u2) to (a2);
		\draw[dashed] (v2) to (b2);
		\draw[dashed] (v2) to (c2);
		\draw[dashed] (w2) to (c2);
		
		\node[ellipse, draw, minimum width=1cm, minimum height=5cm, label=above:$K$] at(\d,2) {};
		\node[ellipse, draw, minimum width=1cm, minimum height=3cm] at(\d+2,3) {};
	\end{tikzpicture}
	\caption{The left side depicts the modulator~$K$ and one clique~$C$ in~$G-K$. Safe edges are depicted with solid lines and unsafe edges are depicted with dashed lines. To reduce visual clutter, we do not show the unsafe edges between two vertices in~$C$. The right side depicts one possible guess for a connecting component. It contains three vertices~$y_1,y_2$, and~$y_3$, some edges between~$K$ and~$\{y_1,y_2,y_3\}$, and a partition of the guessed vertices (in our case~$y_1$ and~$y_2$ are guessed to be safely connected within~$C$). Note that in the graph on the left side, there are many different possibilities to realize the guess on the right side. One possibility is~$y_1=a$,~$y_2=b$ and~$y_3=d$, a second possibility is~$y_1=b$,~$y_2=c$, and~$y_3=d$ and a third possibility~$y_1=a$,~$y_2=c$ and~$y_3=b$.}
	\label{fig:guessInterface}
\end{figure}
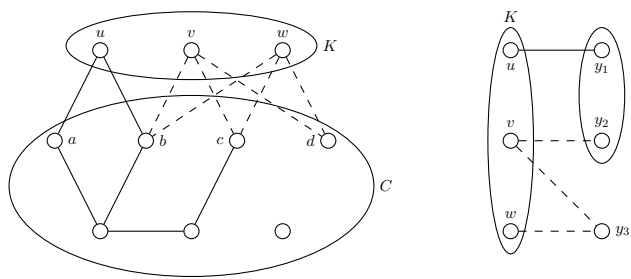%
Finally, we guess for each part of the partition the type of clique that contains the connecting components and how the rest of the clique (that is, all vertices in the clique that are not contained in any connecting component) is connected to the minimal backbone.

We distinguish between the following three types of connections.
To this end, let~$C$ be a clique in~$G'$ that \emph{hosts} at least one connecting component of~$H$ and let~$C' = C \setminus V_H$ be the set of vertices that are not contained in a connecting component.
If there is a connecting component in~$\mathcal{P}$ that contains at least one unsafe edge, then we can replace this edge with a Hamiltonian path through all vertices in~$C'$.
If this is the case or if~$C' = \emptyset$, then we say that~$C'$ is \emph{empty}.
Otherwise, if we can safely connect all vertices in~$C'$ to the backbone using~$|C'|$ edges, then we say that~$C'$ is \emph{efficiently connected} to the backbone.
This is the case if each vertex in~$C'$ is contained in a strong component in~$C$ that contains a) a vertex in some~$P_i \in \mathcal{P}$ or b) a vertex with an incident safe edge to a vertex in~$K$.
If neither of the two cases above applies, then we require at least~$|C'|+1$ edges to connect the vertices in~$C'$ to the backbone.
Note that in this case, we can always find a path through all vertices in~$C'$ and connect the two ends to any vertex in~$C \setminus C'$.
We call this type of connection \emph{inefficient}.


Observe that if two solutions lead to exactly the same set of guesses, then they have the same number of edges\footnote{Assuming that the solutions are minimal, that is, they do not contain edges whose removal yields a smaller solution.} as the difference between the number of edges and vertices in their minimal backbones and the types of the remaining cliques (both of cliques containing parts of the minimal backbone and those which do not) is the same for both solutions.
As argued above, the difference between the number of edges and vertices in these cliques in a solution is completely determined by their type.

It remains to show how to check whether a guess leads to a solution and to analyze the running time.
Towards the former, we first show how to test whether a given clique~$C$ of a guessed type can host a guessed set~$\mathcal{P} = \{P_1,\ldots,P_c\}$ of connecting components such that the rest of the clique has the guessed connection type.
We can handle most combinations of clique type and connection type with a general approach.
One special case has to be treated differently:
The connection type is efficient, each connecting component~$P_i \in \mathcal{P}$ is a usual component which safely connect all respective vertices in~$Y$ and the clique~$C$ is a weak clique.
First, we describe why this case is different from the others.
Then, we show how to handle this special case and how to handle all other cases.
The difference is that if the connection type is ``$C'$ is empty'' or inefficient, then we can pretty much ignore the connection type as we can always greedily find a solution independent of the connecting component.
If one of the connecting components in~$\mathcal{P}$ is a cyclic component or a usual component with at least one unsafe edge in it, then the connection type is ``$C'$ is empty''.\footnote{We may assume that there is only one cyclic component in~$\mathcal{P}$ as we can always merge two cyclic components in one clique into one cyclic component. We need to consider the special case where the cyclic component contains exactly two vertices with edges to vertices in~$K$ as we need to ensure that there is at least one additional vertex not contained in any of the other connecting components in~$\mathcal{P}$.}
The same is true if the clique~$C$ is a singleton.
If~$C$ is a strong clique, then we can also ignore~$C'$ as we can always greedily find an efficient connection independent of the connecting component.
Only if the connection type is efficient, $\mathcal{P}$ only consists of usual components which safely connects all respective vertices in~$Y$, and~$C$ is a weak clique, then we somehow need to ``hit all strong components in~$C$'' which do not have a safe edge to a vertex in~$K$ using the strong components.

The next step in this proof is to describe our algorithm to check if a clique~$C$ can host a set~${\mathcal{P} = \{P_1,P_2,\ldots,P_c\}}$ of connecting components where each~$P_i$ is a usual connecting component consisting of only safe edges.
Informally speaking, we first show that each~$P_i$ is contained in a different strong component in~$C$.
We then use \textsc{Maximum Bipartite Matching} to first check in which strong components~$V_j$ in~$C$ each~$P_i$ can be contained in.
We then check whether we can assign each~$P_i$ to some~$V_j$ such that~$P_i$ can be contained in~$V_j$ and each strong component~$V_j$ which does not have a safe edge to a vertex in~$K$ contains some~$P_i$.
\begin{claim}
\label{claim-special-case}
We can check in~$\Oh(\ell^2 \cdot n^3)$ time whether a clique~$C$ can host a set~${\mathcal{P} = \{P_1,P_2,\ldots,P_c\}}$ of connecting components where each~$P_i$ is a usual connecting component consisting of only safe edges.
\end{claim}

\begin{claimproof}
First, we show that we can assume without loss of generality that each~$P_i$ is contained in a different strong component in~$C$. 
Assume towards a contradiction that there exist two connecting components in~$\mathcal{P}$ in a solution that are contained in the same strong component in~$C$.
Observe that we can assume without loss of generality that each vertex in this strong component is contained in some connecting component in~$\mathcal{P}$.
Hence, there exist two connecting components~$P_i,P_j \in \mathcal{P}$ that are contained in the same strong component and there is a safe edge~$e_S$ between a vertex in~$P_i$ and a vertex in~$P_j$.
Let~$A,B \subseteq K$ be the set of vertices in~$K$ that are connected via~$P_i$ and~$P_j$, respectively.
Since we are considering a valid solution, there has to be some safe connection from some vertex~$a \in A$ to some vertex~$b \in B$.
We built a new solution by merging~$P_i$ and~$P_j$ into one connecting component (by including the edge~$e_S$).
To compensate for this additional edge, we can remove either the edge between~$a$ and a vertex~$c \in P_i$ or the edge between~$b$ and a vertex~$d \in P_j$.
If both of these edges are safe edges or both are unsafe edges, then we can remove either one as there is still a safe (unsafe) connection via the respective other vertex.
If one of the edges is unsafe and the other one is safe, then we remove the unsafe edge.
For the sake of argument, let us say that the edge~$\{a,c\}$ is unsafe and the edge~$\{b,d\}$ is safe.
Removing the edge~$\{a,c\}$ and merging~$P_i$ and~$P_j$ into one connecting component results in a safe connection between~$a$ and~$c$ via safe connections between~$a$ and~$b$ and between~$d$ and~$c$.

We can now check whether~$C$ can host~$\mathcal{P}$ using the textbook~$\Oh(nm)$-time algorithm for \textsc{Maximum Bipartite Matching} twice as follows.
To this end, let~${V_C = \{v_1,v_2,\ldots,v_c\}}$ be the set of vertices in~$C$ and let~${\mathcal{V} = \{V_1,V_2,\ldots,V_q\}}$ be the set of strong components in~$C$.
We first check whether~$P_i$ can be contained in~$V_j$ (for each combination of~$i$ and~$j$) as follows.
Let~$Y_P$ be the set of vertices in~$P_i$ that have guessed edges to vertices in~$K$.
We built a bipartite graph containing a vertex~$y$ for each vertex~$y \in Y_P$ and a vertex~$v$ for each vertex~$v \in V_j$.
There is an edge between~$y$ and~$v$ if and only if~$v$ has safe edges to all vertices in~$K$ that~$y$ has guessed safe edges to and~$v$ has unsafe edges to all vertices in~$K$ that~$y$ has guessed unsafe edges to.
Note that there is a matching in the constructed graph that matches all vertices in~$Y_P$ if and only if there is an assignment of the vertices in~$Y_P$ to distinct vertices in~$V_j$.
Thus, the existence of such a matching determines whether~$P_i$ can be contained in~$V_j$. 
We then check whether we can assign each~$P_i$ to some~$V_j$ such that~$P_i$ can be contained in~$V_j$ and each strong component~$V_j$ which does not have a safe edge to a vertex in~$K$ contains some~$P_i$.
We built a new bipartite graph with one vertex~$u_i$ for each~$P_i \in \mathcal{P}$ and a vertex~$v_j$ for each~$V_j \in \mathcal{V}$.
If~$|\mathcal{P}| > |\mathcal{V}|$, then there is no solution.
Otherwise, we add~$d = |\mathcal{V}| - |\mathcal{P}|$ vertices~$w_1,w_2,\ldots,w_d$.
We then add an edge~$\{u_i,v_j\}$ if~$V_i$ can contain~$P_j$.
Moreover, we add an edge between~$w_i$ and~$v_j$ if~$V_j$ contains a vertex with a safe edge to some vertex in~$K$.
If there exists a perfect matching between the vertices~$v_i$ and the rest of the vertices, then this matching has a one-to-one correspondence to an assignment of each~$P_i$ to some~$V_j$ such that~$P_i$ can be contained in~$V_j$ and each strong component~$V_j$ which does not have a safe edge to a vertex in~$K$ contains some~$P_i$.
Finally, we analyze the running time.
Since we solve at most~$2\cdot n\ell$ instances of \textsc{Maximum Bipartite Matching} on graphs with~$\Oh(n)$ vertices and~$\Oh(n\ell)$ edges in the first stage and, in the second stage, we solve one instance with~$\Oh(n)$ vertices and~$\Oh(n^2)$ edges, the overall running time is in~$\Oh(\ell^2 \cdot n^3)$. 
\end{claimproof}


Next, we present the general algorithm for all remaining cases, that is, the connection type is not efficient, the clique~$C$ is not a weak clique, or the set of connecting components contains an unsafe edge.
Here, we can ignore the connection type and we only need to check whether the guessed set~$\mathcal{P}$ of connecting components can be hosted in a clique~$C$.
Since~$C$ is a clique, any pair of vertices in~$C$ is connected by an edge.
Hence, we can trivially connect any set of~$t$ vertices with unsafe connections using~$t-1$ edges.
We only need to check two~points:
\begin{enumerate}
	\item Is there for each vertex in~$Y$ (the neighbors of vertices in~$K$ in~$C$) a distinct vertex in~$C$ with all the required edges to vertices in~$K$ and
	\item can all sets of vertices that are guessed to be pairwise connected via paths of safe edges inside~$C$ be connected in this way?
\end{enumerate}
Note that the latter point is unfortunately not as simple as checking whether all vertices belong to the same strong component in~$C$ as the example in \cref{fig:ex-safecon} shows.
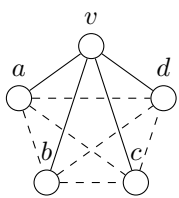
\begin{figure}[t]
	\centering
	\begin{tikzpicture}
		\node[circle,draw, label=$d$] at (0.951,0.309) (d) {};
		\node[circle,draw, label=$v$] at (0,1) (v) {};
		\node[circle,draw, label=$a$] at (-0.951,0.309) (a) {};
		\node[circle,draw, label=$b$] at (-0.587,-0.809) (b) {};
		\node[circle,draw, label=$c$] at (0.587,-0.809) (c) {};
		
		\draw[dashed] (a) -- (b);
		\draw[dashed] (a) -- (c);
		\draw[dashed] (a) -- (d);
		\draw (a) -- (v);
		\draw[dashed] (b) -- (c);
		\draw[dashed] (b) -- (d);
		\draw (b) -- (v);
		\draw[dashed] (c) -- (d);
		\draw (c) -- (v);
		\draw (d) -- (v);
	\end{tikzpicture}
	\caption{A clique with five vertices is shown. Safe edges are depicted with solid lines and unsafe edges are depicted with dashed lines. Suppose we want to connect~$a$ to~$b$ and~$c$ to~$d$ using only safe edges.
	All vertices belong to the same strong component in~$C$ but the only solution is to include all four safe edges in the backbone.
	If each pair of vertices was connected via a safe edge, however, we could connect~$v$ with one of the terminal pairs and directly connect the other terminal pair with a safe edge.
	In this case, we would only use 3 edges.
	Intuitively, we were able to use one edge less by ignoring the connection between the two pairs of terminals.
	Since we assume our guess to be correct, we do not need to connect these terminals inside~$C$ as they will be connected via some paths outside of~$C$.
	Thus, we want each set of vertices that are guessed to be pairwise connected via paths of safe edges inside~$C$ to form their own connected component when considering the graph induced by~$C$ in the solution.}
	\label{fig:ex-safecon}
\end{figure}%
We solve both points as follows.
We built a graph with two vertices~$y,y'$ for each vertex~$y\in Y$ and a vertex~$v$ for each vertex~$v \in C$.
There are edges~$\{y,v\}$ and~$\{y',v\}$ if the following holds.
For each edge~$\{x,y\}$ with~$x \in K$ that is guessed to be in the solution, the edge~$\{x,v\}$ is contained in the input graph~$G$ and the edge~$\{x,v\}$ is safe if and only if the edge~$\{x,y\}$ is guessed to be safe.
Moreover, there is an edge between two vertices~$u,v \in C$ if and only if there is a safe edge between them in~$G$.
Let~$\mathcal{R} = \{R_1,R_2,\ldots,R_r\}$ with~$R_i \subseteq Y$ be a partition of the vertices in~$Y$ according to which vertices are pairwise connected via paths of safe edges.
Let~$R'_i = \{y,y' \mid y \in R_i\}$ be the corresponding vertices in our constructed graph.
As mentioned earlier, we need to consider the special case where one of the connecting components is a cyclic component with exactly two vertices in it as we need to ensure that there is at least one additional vertex that can be included in the cyclic component.
In this case, we add two new vertices~$z,z'$ to the graph, connect each of them to all vertices~$v$ with~$v \in C$, and define the set~$R'_0 = \{z,z'\}$.
We now solve \textsc{Disjoint Connected Subgraphs} where each set~$R'_i$ is one terminal set.
Next, we show that if there is a set of disjoint connected subgraphs each connecting the vertices in one set~$R'_i$, then this corresponds to a solution to both points.
Note that we may again assume that each vertex in~$C$ can only fill the role of one vertex in~$R_i$ and hence selecting any vertex adjacent to~$y$ in the solution gives us a matching between the vertices in~$C$ and the vertices in~$Y$ (any vertex~$y \in R_i$ needs to be connected to at least~$y'$ and hence such a neighbor exists).
Moreover, since the solution is a set of disjoint connected subgraphs and we only included safe edges between vertices in~$C$, we are also ensured that this solution corresponds to a set of connecting components as guessed.
Conversely, if there is a set of disjoint subgraphs in~$C$ that exactly correspond to our set of guesses, then this is also a solution to the instance of \textsc{Disjoint Connected Subgraphs} if we additionally connect each pair of vertices~$y,y'$ to the respective vertex in~$C$.
Thus, we have found a way to check whether~$C$ can host a guessed set~$\mathcal{P}$ of connecting components.


After checking which cliques could potentially host each set~$\mathcal{P}$ of guessed connecting components, it remains to check whether all of these guesses can be fulfilled at the same time.
To this end, we need to check whether all $q \leq p$ sets~$\mathcal{P}_1,\mathcal{P}_2,\ldots,\mathcal{P}_q$ can be hosted each by a distinct clique.
We do this using the textbook~$\Oh(nm)$-time algorithm for \textsc{Maximum Bipartite Matching} as follows.
We build a bipartite graph with one vertex~$v_i$ for each set~$\mathcal{P}_i$ and a vertex~$u_j$ for each connected component (clique)~$C_j$ in~$G'$.
There is an edge~$\{v_i,u_j\}$ in the graph if and only if~$C_j$ can host~$\mathcal{P}_i$.
It then only remains to check whether there is a matching of size~$q$ as in this case each set~$\mathcal{P}_i$ is matched to a distinct clique in~$G'$.

Finally, we analyze the running time of our algorithm.
The algorithm runs in time~$f(\ell) \cdot n^3$ for some computable function~$f$. 
%
Computing~$K$ takes~$\Oh(3^\ell \cdot n^3)$ time and the time for preprocessing all singleton cliques is linear in~$n+m$.
Guessing the edges in~$K$ takes~$\Oh(\ell^{4\ell})$ time as we may assume that we choose at most~$2\ell$ out of~$\ell^2$ edges.
Guessing the number~$p$ of connecting components takes~$\Oh(\ell)$ time and we may assume that~$p = 2\ell$ as any solution can be viewed as a solution with~$p=2\ell$ and a sufficient number of connecting components that do not connect any vertices in~$K$.
The time to guess the structure of one connecting component takes~$\Oh(2^\ell \cdot \ell \cdot (2\ell)^{4\ell} \cdot B(2\ell))$ time, where~$B(i)$ is the~$i$\textsuperscript{th} Bell number.
Hence, the time to guess the structure of~$p$ connecting components is in~$\Oh((2^\ell \cdot \ell \cdot (2\ell)^{4\ell} \cdot B(2\ell))^{2\ell})$.
Next, guessing which connecting components of the minimal backbone are contained in the same clique~$C$ in~$G'$ and what the type and the connection type of~$C$ are take~$\Oh(B(2\ell) \cdot 9^{2\ell})$ time.
Computing a solution for the special case where~$C$ is a weak clique and the connection type is efficient takes~$\Oh(\ell^2 \cdot n^3)$ time by \Cref{claim-special-case}.
The time for all other cases is~$f(4\ell) \cdot n^3$ for each connecting component for some computable function~$f$ \cite{RS95}.
Since there are at most~$2\ell$ connecting components and all other computations (including the bipartite-matching algorithm in the end) only incur a multiplicative factor in~$\ell$ or an additive summand in~$n$, which is dominated by~$n^3$, the overall running time is~$f'(\ell) \cdot n^3$ for some computable function~$f'$.
This concludes the proof.
\end{proof}

We mention that each cluster graph is also a co-graph.
Whether the parameter distance to co-graphs also allows for an FPT-time algorithm remains an interesting open question.

\section{Conclusion}
In this work, we started an investigation into the parameterized complexity of \stue.
Our main results are FPT-time algorithms for a below-upper-bound parameter, the treewidth, and the vertex-deletion distance to cluster graphs, respectively.
Moreover, we give a fairly comprehensive dichotomy between parameters that allow for FPT-time algorithms and those that lead to W[1]-hardness (or in most cases even para-NP-hardness); see \Cref{fig-results} for an overview.

Nonetheless, several open questions remain.
First, what is the status of the parameterizations we were not able to resolve; for example does parameterizing by the distance to cographs or interval graphs allow for FPT-time algorithms?
Second, is there an XP-time algorithm for \stue{} parameterized by the clique-width? 

Third, it would be interesting to study the existence of polynomial kernels for the parameters that yield FPT-time algorithms.
Unfortunately, some of these can be excluded due to the close relation to \HC.
In particular, \HC{} (and hence also \stue{}) does not admit a polynomial kernel with respect to the distance to outerplanar graphs unless \nPHC~\cite{BJK13}. 
Moreover, using the framework of AND-cross compositions, it is not hard to also exclude polynomial kernels 
for the parameters treedepth and bandwidth unless \nPHC:
For treedepth consider the disjoint union of $t$~instances of \HC{} all having $n$~vertices where each edge is made unsafe.
Then, a new vertex~$w$ is added and for each graph~$G_i$ we choose an arbitrary vertex~$u_i\in V(G_i)$ and add a safe edge between~$u_i$ and~$w$.
Finally, we set~$k\coloneqq (n+1)\cdot t$.
For the parameter bandwidth again consider the disjoint union of $t$~instances of \HC{} having $n$~vertices each where each edge is made unsafe.
Next, a path~$W\coloneqq (w_1, \ldots , w_{t})$ of new vertices is added to~$G$.
Each of these edges is safe.
Then, for each graph~$G_i$ we choose an arbitrary vertex~$u_i\in V(G_i)$ and add a safe edge between~$u_i$ and~$w_i$.
Finally, we set~$k\coloneqq (n+2)\cdot t-1$.
However, this still leaves quite a few parameters ready to be investigated in the future.

Last but not least, \stue{} should only be regarded as a first step towards generalizing \textsc{Spanning Tree} to more robust connectivity requirements.
It is interesting to see whether (some of) our positive results can be lifted to the more general problem \textsc{$(p,q)$-Flexible Graph Connectivity}.
Therein, the solution graph should still be~$q$ connected even if up to~$p$ unsafe edges fail. 


\begin{thebibliography}{10}

\bibitem{AHM22}
D.~Adjiashvili, F.~Hommelsheim, and M.~M{\"{u}}hlenthaler.
\newblock Flexible graph connectivity.
\newblock {\em Mathematical Programming}, 192(1):409--441, 2022.

\bibitem{AHMS22}
D.~Adjiashvili, F.~Hommelsheim, M.~M{\"{u}}hlenthaler, and O.~Schaudt.
\newblock Fault-tolerant edge-disjoint $s$-$t$ paths -- beyond uniform faults.
\newblock In {\em Proceedings of the 18th Scandinavian Symposium and Workshops
  on Algorithm Theory ({SWAT}~'22)}, pages 5:1--5:19. Schloss Dagstuhl -
  Leibniz-Zentrum f{\"{u}}r Informatik, 2022.

\bibitem{ANS80}
T.~Akiyama, T.~Nishizeki, and N.~Saito.
\newblock {NP}-completeness of the hamiltonian cycle problem for bipartite
  graphs.
\newblock {\em Journal of Information processing}, 3(2):73--76, 1980.

\bibitem{BMRS17}
J.~Bang{-}Jensen, M.~Basavaraju, K.~V. Klinkby, P.~Misra, M.~S. Ramanujan,
  S.~Saurabh, and M.~Zehavi.
\newblock Parameterized algorithms for survivable network design with uniform
  demands.
\newblock In {\em Proceedings of the 29th Annual {ACM-SIAM} Symposium on
  Discrete Algorithms ({SODA}~'18)}, pages 2838--2850. {SIAM}, 2018.

\bibitem{BY08}
J.~Bang{-}Jensen and A.~Yeo.
\newblock The minimum spanning strong subdigraph problem is fixed parameter
  tractable.
\newblock {\em Discrete Applied Mathematics}, 156(15):2924--2929, 2008.

\bibitem{BCGI22}
I.~Bansal, J.~Cheriyan, L.~Grout, and S.~Ibrahimpur.
\newblock Improved approximation algorithms by generalizing the primal-dual
  method beyond uncrossable functions.
\newblock In {\em Proceedings of the 50th International Colloquium on Automata,
  Languages, and Programming ({ICALP}~'23)}, pages 15:1--15:19. Schloss
  Dagstuhl - Leibniz-Zentrum f{\"{u}}r Informatik, 2023.

\bibitem{BJK13}
H.~L. Bodlaender, B.~M.~P. Jansen, and S.~Kratsch.
\newblock Kernel bounds for path and cycle problems.
\newblock {\em Theoretical Computer Science}, 511:117--136, 2013.

\bibitem{BCHI21}
S.~C. Boyd, J.~Cheriyan, A.~Haddadan, and S.~Ibrahimpur.
\newblock Approximation algorithms for flexible graph connectivity.
\newblock In {\em Proceedings of the 41st {IARCS} Annual Conference on
  Foundations of Software Technology and Theoretical Computer Science
  ({FSTTCS}~'21)}, pages 9:1--9:14. Schloss Dagstuhl - Leibniz-Zentrum
  f{\"{u}}r Informatik, 2021.

\bibitem{CJ22}
C.~Chekuri and R.~Jain.
\newblock Approximation algorithms for network design in non-uniform fault
  models.
\newblock In {\em Proceedings of the 50th International Colloquium on Automata,
  Languages, and Programming ({ICALP}~'23)}, pages 36:1--36:20. Schloss
  Dagstuhl - Leibniz-Zentrum f{\"{u}}r Informatik, 2023.

\bibitem{Civril23}
A.~{\c{C}}ivril.
\newblock A new approximation algorithm for the minimum 2-edge-connected
  spanning subgraph problem.
\newblock {\em Theoretical Computer Science}, 943:121--130, 2023.

\bibitem{CFK+15}
M.~Cygan, F.~V. Fomin, L.~Kowalik, D.~Lokshtanov, D.~Marx, M.~Pilipczuk,
  M.~Pilipczuk, and S.~Saurabh.
\newblock {\em Parameterized Algorithms}.
\newblock Springer, 2015.

\bibitem{DF13}
R.~G. Downey and M.~R. Fellows.
\newblock {\em Fundamentals of Parameterized Complexity}.
\newblock Springer, 2013.

\bibitem{FML22}
A.~E. Feldmann, A.~Mukherjee, and E.~J. van Leeuwen.
\newblock The parameterized complexity of the survivable network design
  problem.
\newblock In {\em Proceedings of the 5th Symposium on Simplicity in Algorithms
  ({SOSA}~'22)}, pages 37--56. {SIAM}, 2022.

\bibitem{FLL03}
C.~Feremans, M.~Labb{\'{e}}, and G.~Laporte.
\newblock Generalized network design problems.
\newblock {\em European Journal of Operational Research}, 148(1):1--13, 2003.

\bibitem{FGLS10}
F.~V. Fomin, P.~A. Golovach, D.~Lokshtanov, and S.~Saurabh.
\newblock Intractability of clique-width parameterizations.
\newblock {\em {SIAM} Journal on Computing}, 39(5):1941--1956, 2010.

\bibitem{FGLSZ19}
F.~V. Fomin, P.~A. Golovach, D.~Lokshtanov, S.~Saurabh, and M.~Zehavi.
\newblock Clique-width {III:} {H}amiltonian cycle and the odd case of graph
  coloring.
\newblock {\em {ACM} Transactions on Algorithms}, 15(1):9:1--9:27, 2019.

\bibitem{GG12}
H.~N. Gabow and S.~Gallagher.
\newblock Iterated rounding algorithms for the smallest \emph{k}-edge connected
  spanning subgraph.
\newblock {\em {SIAM} Journal on Computing}, 41(1):61--103, 2012.

\bibitem{GJ79}
M.~R. Garey and D.~S. Johnson.
\newblock {\em Computers and Intractability: {A} Guide to the Theory of
  NP-Completeness}.
\newblock W. H. Freeman, 1979.

\bibitem{Golumbic04}
M.~C. Golumbic.
\newblock {\em Algorithmic graph theory and perfect graphs}.
\newblock Academic Press, 1980.

\bibitem{Huh04}
W.~T. Huh.
\newblock Finding 2-edge connected spanning subgraphs.
\newblock {\em Operations Research Letters}, 32(3):212--216, 2004.

\bibitem{HVV19}
C.~Hunkenschr{\"{o}}der, S.~S. Vempala, and A.~Vetta.
\newblock A 4/3-approximation algorithm for the minimum 2-edge connected
  subgraph problem.
\newblock {\em {ACM} Transactions on Algorithms}, 15(4):55:1--55:28, 2019.

\bibitem{IPZ01}
R.~Impagliazzo, R.~Paturi, and F.~Zane.
\newblock Which problems have strongly exponential complexity?
\newblock {\em Journal of Computer and System Sciences}, 63(4):512--530, 2001.

\bibitem{Kloks94}
T.~Kloks.
\newblock {\em Treewidth: Computations and Approximations}.
\newblock Springer, 1994.

\bibitem{Kor22}
T.~Korhonen.
\newblock A single-exponential time 2-approximation algorithm for treewidth.
\newblock In {\em Proceedings of the 62nd Annual Symposium on Foundations of
  Computer Science (FOCS~'22)}, pages 184--192. IEEE, 2022.

\bibitem{K56}
J.~B. Kruskal.
\newblock On the shortest spanning subtree of a graph and the traveling
  salesman problem.
\newblock {\em Proceedings of the American Mathematical Society}, 7(1):48--50,
  1956.

\bibitem{MR99}
M.~Mahajan and V.~Raman.
\newblock Parameterizing above guaranteed values: Maxsat and maxcut.
\newblock {\em Journal of Algorithms}, 31(2):335--354, 1999.

\bibitem{MRS09}
M.~Mahajan, V.~Raman, and S.~Sikdar.
\newblock Parameterizing above or below guaranteed values.
\newblock {\em Journal of Computer and System Sciences}, 75(2):137--153, 2009.

\bibitem{P57}
R.~C. Prim.
\newblock Shortest connection networks and some generalizations.
\newblock {\em The Bell System Technical Journal}, 36(6):1389--1401, 1957.

\bibitem{R39}
H.~E. Robbins.
\newblock A theorem on graphs, with an application to a problem of traffic
  control.
\newblock {\em The American Mathematical Monthly}, 46(5):281--283, 1939.

\bibitem{RS95}
N.~Robertson and P.~D. Seymour.
\newblock Graph minors {XIII}: {T}he disjoint paths problem.
\newblock {\em Journal of Combinatorial Theory, Series B}, 63(1):65--110, 1995.

\bibitem{SV14}
A.~Seb{\"{o}} and J.~Vygen.
\newblock Shorter tours by nicer ears: 7/5-approximation for the graph-{TSP},
  3/2 for the path version, and 4/3 for two-edge-connected subgraphs.
\newblock {\em Combinatorica}, 34(5):597--629, 2014.

\bibitem{LSDC06}
L.~V. Snyder, M.~P. Scaparra, M.~S. Daskin, and R.~L. Church.
\newblock Planning for disruptions in supply chain networks.
\newblock {\em INFORMS TutORials in Operations Research}, pages 234--257, 2006.

\bibitem{Tak90}
L.~Tak{\'{a}}cs.
\newblock On {C}ayley's formula for counting forests.
\newblock {\em Journal of Combinatorial Theory, Series A}, 53(2):321--323,
  1990.

\bibitem{West00}
D.~B. West.
\newblock {\em Introduction to Graph Theory}.
\newblock Prentice Hall, 2000.

\bibitem{XOX02}
Y.~Xu, V.~Olman, and D.~Xu.
\newblock Clustering gene expression data using a graph-theoretic approach:
  {A}n application of minimum spanning trees.
\newblock {\em Bioinformatics}, 18(4):536--545, 2002.

\end{thebibliography}
\end{document}